\documentclass[aps,reprint,groupedaddress,superscriptaddress,longbibliography]{revtex4-1}

\usepackage{qcircuit}
\usepackage{enumitem}
\usepackage{amsmath}
\usepackage{amssymb}
\usepackage{amsthm}
\usepackage{amsfonts}
\usepackage{braket}
\usepackage{dsfont}
\usepackage{textcomp,gensymb}
\usepackage{bm}

\usepackage{color}
\usepackage[dvipsnames]{xcolor}
\usepackage{tikz}
\usepackage{booktabs}
\usepackage{multirow}
\usepackage{nicefrac}

\usepackage{graphicx, import}
\usepackage{soul}
\usepackage[utf8]{inputenc}
\usepackage[export]{adjustbox}
\usepackage[caption=false]{subfig}
\usepackage{parskip} 
\usepackage{hyperref}
\usepackage{url}

\newtheorem{definition}{Definition}
\newtheorem{theorem}{Theorem}
\newtheorem{corollary}{Corollary}[theorem]
\newtheorem{lemma}[theorem]{Lemma}

\newcommand{\schroedinger}{Schr\"{o}dinger }

\renewcommand{\H}{\hat{H}}
\renewcommand{\P}{\hat{P}}
\newcommand{\x}{\hat{x}}
\newcommand{\p}{\hat{p}}
\newcommand{\q}{\hat{q}}
\newcommand{\bfx}{\mathbf{\x}}
\newcommand{\y}{\hat{y}}
\newcommand{\z}{\hat{z}}
\newcommand{\bfp}{\mathbf{\p}}
\newcommand{\bfR}{\mathbf{R}}
\renewcommand{\a}{\hat{a}}
\newcommand{\aD}{\hat{a}^\dagger}
\newcommand{\hil}{\mathcal{H}}
\newcommand{\I}{\mathds{1}}
\newcommand{\Itwo}{I\textsubscript{2} }
\newcommand{\Itwonospace}{I\textsubscript{2}}

\newcommand{\termplus}{\frac{\sqrt{3}+1}{4\sqrt{2}}}
\newcommand{\termminus}{\frac{\sqrt{3}-1}{4\sqrt{2}}}

\DeclareMathOperator*{\cover}{Cover}
\DeclareMathOperator*{\var}{Var}
\DeclareMathOperator*{\cov}{Cov}
\DeclareMathOperator*{\bin}{bin}
\DeclareMathOperator*{\partition}{Part}
\DeclareMathOperator*{\tr}{Tr}

\newcommand{\Cvv}{C_{\parallel,\parallel}}
\newcommand{\Cvh}{C_{\parallel, \perp}}
\newcommand{\Chh}{C_{\perp, \perp}}

\def\ketbra#1#2{{\vert#1\rangle\!\langle#2\vert}}

\begin{document}
\title{Coarse-grained intermolecular interactions on quantum processors}

\author{Lewis~W.~Anderson} 
\email{lewis.anderson@physics.ox.ac.uk}
\affiliation{Clarendon Laboratory, University of Oxford, Parks Road, Oxford OX1 3PU, UK}

\author{Martin~Kiffner} 
\affiliation{Centre for Quantum Technologies, National University of Singapore, 3 Science Drive 2, Singapore 117543}
\affiliation{Clarendon Laboratory, University of Oxford, Parks Road, Oxford OX1 3PU, UK}

\author{Panagiotis~Kl.~Barkoutsos}
\affiliation{IBM Quantum, IBM Research Zurich, Säumerstrasse 4, 8803 Rüschlikon, Switzerland}

\author{Ivano~Tavernelli}
\affiliation{IBM Quantum, IBM Research Zurich, Säumerstrasse 4, 8803 Rüschlikon, Switzerland}

\author{Jason~Crain}
\affiliation{IBM Research Europe, The Hartree Centre STFC Laboratory, Sci-Tech Daresbury, Warrington WA4 4AD, UK}
\affiliation{Department of Biochemistry, University of Oxford, Oxford OX1 3QU, UK}

\author{Dieter~Jaksch}
\affiliation{Institut f\"ur Laserphysik, Universit\"at Hamburg, 22761 Hamburg, Germany}
\affiliation{Clarendon Laboratory, University of Oxford, Parks Road, Oxford OX1 3PU, UK}
\affiliation{Centre for Quantum Technologies, National University of Singapore, 3 Science Drive 2, Singapore 117543}

\date{\today}

\begin{abstract}
Variational quantum algorithms (VQAs) are increasingly being applied in simulations of strongly-bound (covalently bonded) systems using full molecular orbital basis representations. The application of quantum computers to the weakly-bound intermolecular and non-covalently bonded regime however has remained largely unexplored. In this work, we develop a coarse-grained representation of the electronic response that is ideally suited for determining the ground state of weakly interacting molecules using a VQA. We require qubit numbers that grow linearly with the number of molecules and derive scaling behaviour for the number of circuits and measurements required, which compare favourably to traditional variational quantum eigensolver methods. We demonstrate our method on IBM superconducting quantum processors and show its capability to resolve the dispersion energy as a function of separation for a pair of non-polar molecules - thereby establishing a means by which quantum computers can model Van der Waals interactions directly from zero-point quantum fluctuations. Within this coarse-grained approximation, we conclude that current-generation quantum hardware is capable of probing energies in this weakly bound but nevertheless chemically ubiquitous and biologically important regime. Finally, we perform experiments on simulated and real quantum computers for systems of three, four and five oscillators as well as oscillators with anharmonic onsite binding potentials; the consequences of the latter are unexamined in large systems using classical computational methods but can be incorporated here with low computational overhead.
\end{abstract}

\maketitle


\section{Introduction}

Computer simulation of matter at the atomic and molecular scale has advanced to the point where it is now playing an increasingly significant role in the design and discovery of new functional materials.  However, these simulations all face the common challenge of how to strike inevitable compromises between the physical realism of an underlying model and the practicality of its implementation across the deep hierarchy of intermolecular forces responsible for cohesion in condensed phases.

This hierarchy spans a continuum of energy scales ranging from greater than 1~eV for \textit{intra-}molecular bonds, extending down to $10^{-1}-10^{-2}$~eV~ for the inter-molecular regime~\cite{saito2020chemical}.  The weakest of these are the dispersion interactions arising from instantaneous quantum fluctuations of the electron distribution. These give rise to correlated induced multipole moments with the leading-order contribution coming from dipolar coupling -- attenuated by the familiar  $1/R^6$ asymptotic decay with distance. Dispersion interactions therefore form part of the non-local correlation energy of the system and are not pair-wise additive. Although generally weaker than electrostatic interactions or hydrogen bonds, their  long-ranged character and ubiquity means  that they exert influence to some extent in all non-covalent associations. 

There is now clear evidence that dispersion effects extend beyond cohesive energy and have impact on structural and mechanical properties~\cite{reilly2015van, shtogun2010many} as well as barrier heights and phase transitions~\cite{pastorczak2017intricacies}. They can therefore play governing roles in a wide range of condensed matter phenomena including liquid state physics, molecular crystal stability, low-dimensional and layered system stacking, hybrid organic–inorganic interface properties and biophysical interactions such as drug-ligand binding.  Moreover, as these interactions scale with system size, their importance is expected to grow for larger molecules and supramolecular assemblies~\cite{hermann2017first}. 

Principally, there are two distinct strategies employed to capture this molecular interaction hierarchy in computer models: Firstly, there are empirical potential or force-field methods. Here the system energy is typically decomposed into  bonded  (stretches, bends and torsions) and non-bonded (electrostatic and dispersive) terms. These are parameterised to fit certain experimental reference quantities. The motion of the system can then  be evolved by classical mechanics (molecular dynamics), and trajectories averaged to obtain static properties. These methods have the advantage of extreme computational efficiency meaning that they can be applied to relatively large systems. However they generally lack explicit electronic responses and, consequently, transferability to thermodynamic or environmental conditions outside those of the parameterisation range is limited. 

Secondly, there are so-called first principles approaches in which the Schrodinger equation is solved to various levels of approximation. Here the prevailing methods for molecular dynamics applications are based on Kohn-Sham density functional theory. In these models, there are no empirical force laws: Instead, forces are computed “on-the-fly” directly from the electronic structure according to the Hellmann-Feynman~\cite{hellman1937einfuhrung, feynman1939forces} theorem with trajectories evolved on the Born-Oppenheimer surface~\cite{born1927quantentheorie}.  These techniques are computationally resource-intensive and tend to be limited to relatively small systems. Also, common density functional approximations do not properly capture non-local correlations involved in dispersion forces and there are active efforts to devise dispersion-corrected density functionals~\cite{french2010long,von2004optimization,  tkatchenko2009accurate, grimme2011density}. High-level quantum chemistry methods (such as configuration interaction (CI)~\cite{slater1929theory, boys1950electronic, szalay2012multiconfiguration} or canonical coupled cluster (CC)~\cite{vcivzek1966correlation, bartlett2007coupled} with perturbative - single, double, triple - excitations) can account for dispersion interactions explicitly within various levels of approximation albeit with even more severe limitations imposed by system size scaling~\cite{pieniazek2007benchmark, oliveira2017systematic, bistoni2020finding}.

Quantum computing offers promising routes for simulations of correlated electronic systems at the first-principles level. Materials modelling at atomic and molecular scales is therefore one of its most anticipated application areas. One such route is to use variational quantum algorithms (VQAs)~\cite{preskill2018quantum,cerezo2021variational} for solving the \schroedinger equation for chemical systems~\cite{mcardle2020quantum}. VQAs have the benefits of requiring relatively short coherence times and modest numbers of qubits which is achieved by using a shallow parameterised quantum circuit embedded within a classical optimisation routine to find the ground state of a given problem Hamiltonian~\cite{peruzzo2014variational}. Since the first proposal to use variational methods namely, the variational quantum eigensolver (VQE) for chemical systems~\cite{peruzzo2014variational}, VQAs have seen application to a significant number of areas including combinatorial optimisation~\cite{farhi2014quantum,farhi2016quantum}, strongly correlated materials~\cite{bauer2016hybrid,kreula2016few,kreula2016non} and non-linear partial differential equations~\cite{lubasch2020variational}. Current capabilities to model molecular properties from first principles in large systems using VQAs are, however, limited by the prohibitive computational cost of the high-level, wavefunction methods typically deployed in quantum chemistry. For example, solution of the Schrödinger equation at the full configuration interaction (so-called “full CI”) level for the ground-state wave function for molecular systems exhibits factorial scaling~\cite{ollitrault2020quantum}.

Chemical applications of VQAs have consequently been confined to the strongly-bound intramolecular regime of the interaction hierarchy and to small systems with examples now demonstrated for dissociation curves of small molecules~\cite{grimsley2019adaptive, kandala2017hardware}, ground state energies for intramolecular bonds~\cite{nam2020ground}, molecular excitation energies~\cite{ollitrault2020quantum}, covalent binding energies of hydrogen chains, isomerisation barriers~\cite{google2020hartree}, molecular vibrations~\cite{mcardle2019digital, ollitrault2020hardware} and, two-site DMFT calculations~\cite{rungger2019dynamical, jaderberg2020minimum}. All except the last two examples consider Hamiltonians and states within full electronic basis sets such as Slater-type or Gaussian-type orbitals~\cite{helgaker2014molecular}. For the sake of comparison with our method, we will refer to these as orbital-based VQE methods.

\begin{figure}
\centering
\includegraphics{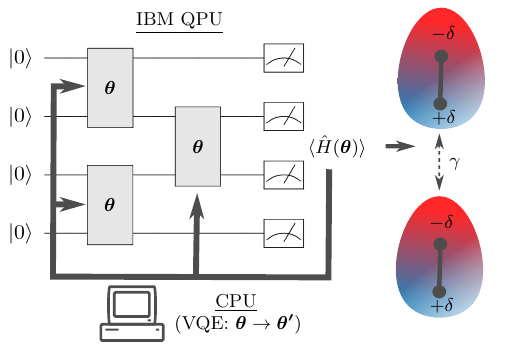}
\caption{Schematic of the VQE optimisation procedure and output states showing instantaneous asymmetry of the molecular charge distribution. These correlated charge density fluctuations lead to dispersion interaction which, to the lowest order, is the familiar $1/R^6$ London attraction.}
\label{fig: abstract}
\end{figure}

In contrast, the intermolecular (non-covalent) interaction regime has been largely unexplored using quantum algorithms and is a gap which will need to be filled in order to address realistic material science problems using near-term, noisy, quantum computers. 
Due to the size of the molecular complexes of interest and the necessary inclusion of core electrons for dispersive forces, the required number of electrons and orbitals implies a large number of basis functions and qubits~\cite{sherrill1996computational,sherrill2017wavefunction}.
This makes the direct scale up of current quantum algorithms for electronic structure calculations very impractical and therefore strategies are needed to reduce the number of electrons and orbitals treated explicitly at the highest levels of accuracy.

To address this difficulty, we develop a coarse-grained model of electronic interactions that is ideally suited for VQA and can easily be extended to represent a wide range of possible interactions. Our approach is inspired by and extends a maximally coarse-grained model used in classical computations~\cite{cipcigan2019electronic}. In this model the responsive (polarisable) part of the molecular charge distribution is modelled as a quantum harmonic oscillator embedded in the molecular frame and its properties can be tuned to reproduce reference values of polarisability or dispersion coefficients for real molecular species. Zero-point multipolar fluctuations are present by construction and so dispersion arises naturally. While the model parameters are defined empirically to fit molecular properties there are no interaction potentials or force laws between atoms or molecules defined a priori. These are instead determined from the coarse-grained electronic structure as in ab-initio approaches.

This model blends therefore features from both empirical potential and first principles methodologies and has already been solved, with model limitations, using a path integral approach on classical processors. In particular, it has been successfully applied to noble gas liquids and solids (pure dispersion) and to water (where both polarisation and dispersion are present) across its phase diagram where it has correctly predicted vaporisation enthalpy, temperature of maximum density, temperature dependences of the dielectric constant and surface tension~\cite{cipcigan2019electronic,cipcigan2018structure,sokhan2015signature} and revealed molecular signatures of Widom line crossing in the supercritical fluid~\cite{sokhan2015molecular}. 
These classical methods have however been developed with the specific aim of deriving force laws and are restricted to harmonic confining potentials where certain order interactions are missing and others are related by simple algebraic relationships that hold only approximately for real systems~\cite{cipcigan2016electronic}. 

By using quantum processors to model a coarse-grained representation, our method gives a more general framework than classical approaches. This allows for the extension to anharmonic oscillators, which would capture these missing interactions as well as allowing for the introduction of more complex physics through non-linear interaction terms. Fig.~\ref{fig: abstract} shows a schematic representation of the algorithm and correlated charge fluctuations captured by the model.

Here we summarize our main findings:
Firstly, the number of qubits required to simulate the model on a quantum computer grows linearly with the number of molecules. Secondly, the number of unique circuits required to measure our Hamiltonian with all-to-all dipolar interactions scales at most linearly with the number of oscillators. This compares favourably to the $O(N^3)$ unique circuits that must be measured for a full orbital-based VQE Hamiltonian with $N$ orbitals~\cite{gokhale2020n}. For our algorithm, relaxation of the harmonic approximation incurs negligible experimental overhead. Classical methods, such as path integral and Monte Carlo techniques, typically rely on the efficient sampling~\cite{westbroek2018user} of two-point correlators for Gaussian states generated by harmonic Hamiltonians. However, the non-Gaussian ground states of anharmonic Hamiltonians are much more costly to sample. Thus, anharmonic potentials have not been studied much using classical computational approaches.

\begin{figure}
\includegraphics{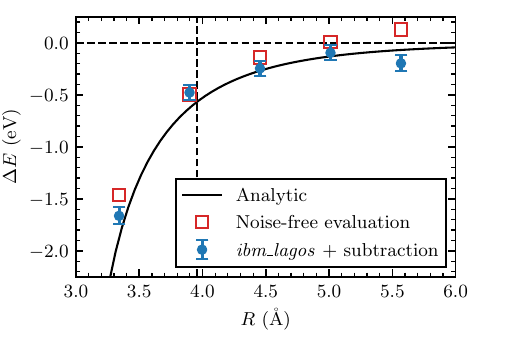}
\caption{London dispersion energy ($\Delta E$) calculated using variational quantum algorithm for one-dimensional quantum Drude oscillators computing interacting I\textsubscript{2} dimers. The black line gives the dispersion energy obtained analytically. Red squares correspond to noise-free evaluation (via state vector simulations) of the experimentally optimized coupled ground state.
Blue circles correspond to energy directly obtained from experimentally optimized ground states (error bars represent one-standard deviation on the mean). The horizontal line marks zero dispersion interaction corresponding to uncoupled oscillators at $R=\infty$ and vertical dashed line corresponds to the I\textsubscript{2} van der Waals diameter.}\label{fig: spectrum}
\end{figure}

We also present proof of principle experiments using IBM superconducting quantum processors demonstrating the solubility of our model on current generation hardware. We focus on two linear non-polar molecules where the electronic responses are described by embedded quantum oscillators. Here the interactions arise from London dispersion forces for which the leading contribution comes from dipolar quantum fluctuations. Looking ahead to the main experimental results of this paper, Fig.~\ref{fig: spectrum} shows the result of our VQA; combined with a simple method to extract the binding energy; for a pair of these interacting oscillators with parameters chosen to represent an Iodine (\Itwonospace) molecular dimer. This demonstration successfully recreates the leading order London dispersion interaction, thereby giving the first demonstration of the use of quantum processors to model dispersion forces directly arising from quantum fluctuations. 

Finally, we extend our proof of principle experiments to systems of three, four and five oscillators using simulated quantum computers. We show how the algorithm performs as a function of ansatz size and demonstrate that it is able to find groundsates of systems containing non-trivial many body interactions. We also perform experiments using an IBM superconducting quantum processor in order for a system including anharmonic onsite binding potentials. We show that the algorithm is again able to capture systems showing strong anharmonic interactions however device noise levels must be further improved in order to accurately measure the anharmonic contribution to the energy of the system, even when using state-of-the-art error mitigation techniques.

Our coarse-grained model allows for the use of quantum computers to efficiently model (with favourable scaling) the inter-molecular interactions between large molecules, which are important in characterising a wide range of chemical and biological systems. This fills an essential gap in the chemical interaction hierarchy and allows us to extend current coarse-grained classical methods to capture a richer set of physical behaviours with the necessary accuracy. 
Given the reduced resource requirements (i.e., number of qubits, gate operations and circuit measurements) when compared to full orbital-based electronic structure calculations, the methodology presented here may be more readily scalable to demonstrating an application of quantum advantage.

The layout of this article is as follows. In Sec.~ \ref{sec: model on qc} we describe the coarse-grained model and how it can be solved using a variational quantum algorithm. This includes: Sec. \ref{subsec: coarse graining} describing the coarse-grained model of interacting oscillators; Sec. \ref{subsec: algorithm rep} giving a description of how the this system is represented on a quantum computer; Sec. \ref{subsec: shots} which makes use of results from graph theory and measure theory to give a scheme for grouping operators needed to measure the Hamiltonian as well as derive upper and lower bounds for the number of shots required to measure the Hamiltonian to a given accuracy; and Sec. \ref{subsec: ansatz} which describes the variational ansatz used in our experiments. In Sec. \ref{sec: ibmq dispersion}, we present proof of principle results modelling the London dispersion interaction of an \Itwo dimer and C technique to subtract the effect of device noise to accurately measure the binding energy. In Sec. \ref{sec: model extensions} we describe how the one-dimensional model can be extended to three dimensions and to include anharmonic and non-linear terms to model richer physics than those accessible to efficient classical methods. In Sec. \ref{sec: further results} we present further experiments on real and simulated quantum processors. This includes Sec. \ref{subsec: anharmonic experiments} for systems of anharmonic oscillators; as well as Sec. \ref{subsec: many oscillators} where we increase the size of the harmonic system to include three, four and five oscillators. In Sec. \ref{sec: discussion} we discuss our results and describe further research directions.


\section{Coarse-grained modelling on digital quantum computers}\label{sec: model on qc}
\subsection{Coarse-grained model}\label{subsec: coarse graining}

In the coarse-grained model, collective electronic responses are described by a charged quantum Drude oscillator (QDO) tethered to an atomic nucleus or embedded in a molecular frame. For neutral species the Drude (quasi-)particle and the nucleus to which it is bound have charges $-q$ and $+q$, respectively. The correlated displacements of these QDOs, which correspond to instantaneous charge density fluctuations in molecules, allows us to capture dispersive interactions between non-polar molecules. It is an electronically coarse-grained model which nonetheless exhibits the full hierarchy of electronic responses at long range: 
Many-body polarisation and dispersion to all orders arise naturally and, in principle, intermolecular forces can then be computed ``on the fly'' from the coarse-grained electron distribution of Coulomb-coupled oscillators without the need for {\it a priori} force laws or empirical interaction potentials. 
One further feature of this model, is that it obeys distinguishable particle statistics -- Fermi statistics are not required. The model must therefore be augmented by a simple short ranged repulsion term. 
The basic method has been demonstrated on classical processors in the case of noble gas solids and single-component liquids~\cite{jones2013electronically, jones2013quantum}. Using only the properties of the isolated molecule in the model parameters, emergent properties of these systems have been predicted across a range of their respective phase diagrams. For a recent review of these methods and results see~\cite{cipcigan2019electronic}.

For simplicity, for most of this work we will focus on the case of one-dimensional QDOs with dipole order coupling, however much of the methodology developed here can be applied straightforwardly to systems of three-dimensional oscillators and we give a brief description of this extension in  Sec.~\ref{subsec: 3d}. Even the simple one-dimensional model recreates the London dispersion interaction proportional to $R^{-6}$ and $\alpha^2$, where $\alpha$ is the dipole polarisability and $R$ the physical separation, up-to a constant factor for spherical non-polar molecules and exactly for a pair of linear non-polar molecules in the limit of infinitely anisotropic polarisability (See Appendix~\ref{app: one-dimensional oscillators} for further details.) 

The non-dimensional Hamiltonian for $N$ one-dimensional QDOs with quadratic coupling is
\begin{equation}\label{eqn: hamiltonian}
\H = \sum_{i=1}^N (\x_i^2 + \p_i^2) + \sum_{j>i} \gamma_{i,j} \x_i\x_j,
\end{equation}
where $\x_i$ and $\p_i$ are the bosonic position and momentum operators for oscillator $i$ and $\gamma_{i,j}\in \mathbb{R}$ is the coupling between oscillators  $i$ and $j$. Throughout this work we will implicitly work in energy units of $\hbar \omega/2$, for simplicity of notation assuming that all oscillators share the same trapping frequency $\omega$. For each pair of identical one-dimensional oscillators aligned parallel to each other and along the inter-oscillator axis, separated by distance $R_{i,j}$, the coupling constant is given by $\gamma_{i,j}=-2(q^2/\mu\omega^2) R_{i,j}^{-3}=-4\alpha R_{i,j}^{-3}$, where the effective charge of the Drude particle $q$, the effective mass $\mu$ and the oscillator frequency $\omega$ are model parameters. The dipole polarisability is given by $\alpha=q^2/\mu\omega^2$. If instead the oscillators are aligned parallel to each other and perpendicular to the inter-oscillator axis, the coupling is given by $\gamma_{i,j}=2\alpha R_{i,j}^{-3}$.

One advantage of this harmonically bound coarse-grained model is that efficient classical methods exist which can be used for benchmarking quantum implementations. However, the price paid for this harmonic approximation is that the relative strengths of interactions within the model are governed by Gaussian statistics meaning that there are simple algebraic relations relating interactions at different orders. 
This leads to the absence of terms, such as the second hyperpolarisability (dipole-dipole-dipole-dipole)~\cite{jones2013quantum}, due to certain symmetries that emerge in the harmonic model.

In what follows we present a variational quantum algorithm for calculating intermolecular interactions between molecules using a one-dimensional QDO model and outline how to extend this to a full three dimensional representation. 
Such a quantum algorithm is not limited to the inclusion of harmonic Hamiltonian terms and can be extended to include anharmonic trapping potentials as well as non-linear interaction terms that may allow for a richer set of physics beyond Gaussian statistics and the study of more realistic chemical systems~\cite{jones2013quantum}.


\subsection{Algorithm description} \label{subsec: algorithm rep}

In order to find the ground state of the Hamiltonians of interest we use the common procedure for a Variational Quantum Eigensolver (VQE) algorithm \cite{peruzzo2014variational}. The main goal of VQE is to produce quantum states $\ket{\psi(\bm{\theta})}$ using a quantum circuit parameterised by real parameters $\bm{\theta} = \set{\theta_0, \theta_1,\dots}$ which can be varied. Throughout the algorithm, the parameters are updated using a classical feedback loop in order to minimise the expected values of the Hamiltonian on the particular quantum state. Guided from the variational principle the optimisation procedure can provide a quantum state that approximates the exact ground state of the system by solving the minimisation problem

\begin{equation}
\min_{\bm{\theta}} \braket{\psi(\bm{\theta})|\hat{H}|\psi(\bm{\theta})},
\end{equation}

provided the variational ansatz is able to represent a state close to the ground state within the required precision. In the following we outline how we encode the states of harmonic QDOs into qubit states and how the Hamiltonian is measured efficiently on a quantum computer before defining our variational ansatz later.

We represent the state of the QDOs in the Fock basis which for Bosonic oscillators form an infinite basis. We are limited to a finite number of qubits and so must restrict ourselves to the space of a finite subset of Fock states. This will naturally lead to solutions that approximately match the true states in full Fock space. Since dispersive interaction energies are expected to be much smaller than the bare energy scale $\hbar\omega$, the groundstates of realistic Hamiltonians will be dominated by the lowest energy Fock states and thus the truncation of the Fock space will not severely impact the accuracy of solutions.
We use a compact encoding~\cite{sawaya2020resource}, which was proposed in earlier work on vibrational modes on a quantum computer \cite{dumitrescu2018cloud,mcardle2019digital}. For this encoding, the first $d$ Fock states are mapped onto $m=\log_2d$ qubits as

\begin{equation}
\hil_{d-\text{Fock}} \rightarrow \hil_2^{\otimes m}: \; \ket{\underline{n}} \rightarrow \ket{\text{bin}_m(n)}
\end{equation}

where $\text{bin}_m(n)$ denotes the $m$-length binary representation of integer $n \in \set{0,1,\dots, d-1}$. Throughout this work,  underlined integers within a $\ket{\cdot}$ will denote a Fock state and all other states will be written in the computational basis of the qubits. The basis for multiple QDOs is then given by the tensor product of the single QDO basis such that for a total of $N$ $2^m$-dimensional QDO we will use $M=mN$ qubits.

Although we are restricted to the lowest $d$ Fock states, we expect that the ground states required will be sufficiently close to the uncoupled ground state $\ket{\underline{0}}^{\otimes N}$ such that a finite number of basis states will give a good approximation to the ground state energy, particularly in the weakly-coupled regime. Since we require logarithmically many qubits, this qubit encoding will be efficient provided the number of Fock states required for a chosen accuracy grows less than double-exponentially.

To measure the cost function for a given variational state we must decompose each term in the Hamiltonian in \eqref{eqn: hamiltonian}, restricted to a finite Fock basis, into tensor product Pauli operators $\hat{P} \in \set{\I, X, Y, Z}^{\otimes M}$ that can be measured on a quantum computer. For the non-interacting and coupling terms we show this for a single QDO and a single pair of QDOs respectively and note that the same decomposition can be used for all terms within the Hamiltonian just acting on different subsets of $m$ and $2m$ qubits. For the non-interacting contribution of a single QDO, there are $m$ non-identity $Z$ operators that must be measured:

\begin{equation} \label{eqn: nonint paulis}
\x^2 + \p^2 = 2^m \I^{\otimes m} - \sum_{i=0}^{m-1} 2^i Z_i,
\end{equation}

where $Z_i$ is the Pauli Z operator acting on qubit $i$. For a pair of $m$ qubit QDOs the coupled part of the Hamiltonian leads to $\left(\frac{1}{2}d\log_2d\right)^2$ Pauli terms. The form of the coupling terms for arbitrary $d$, as well as a proof for the number of terms, is given in Appendix \ref{app: pauli decomp}. As an example, the decompositions for $m=1$ and $m=2$ are given by 
\begin{widetext}
\begin{equation}
\x_1\x_2\rvert_{m=1} = X_1X_2,
\end{equation}
\begin{equation}
\begin{split}
\x_1\x_2 \lvert_{m=2} = & \frac{\sqrt{3}+2}{4}X_1X_3  + \termplus X_1X_2X_3 + \termplus X_1X_3X_4 +\frac{1}{4}X_1X_2X_3X_4 + \termplus Y_1Y_2X_3 +\frac{1}{4}Y_1Y_2X_3X_4 \\
    & + \termplus X_1Y_3Y_4 +\frac{1}{4}X_1X_2Y_3Y_4 - \termminus Y_1Y_2X_3Z_4 -\termminus X_1Z_2Y_3Y_4 - \frac{1}{4}X_1Z_2X_3 \\ &- \termminus X_1Z_2X_3X_4 - \frac{1}{4} X_1X_3Z_4 - \termminus X_1X_2X_3Z_4 + \frac{2-\sqrt{3}}{4}X_1Z_2X_3Z_4 +\frac{1}{4}Y_1Y_2Y_3Y_4.
\end{split}
 \label{eqn: coupling paulis}
\end{equation}
\end{widetext}
The interaction Hamiltonian for $N$ QDOs can then be constructed by combining these pair-wise interactions for all $N(N-1)/2$ pairs of QDOs with the correct coupling constant.


\subsection{Measurement cost} \label{subsec: shots}

\begin{table*}
\begin{ruledtabular}
\begin{tabular} {@{}ccccc@{}}
\multirow{2}{*}{Measure over states} & \multicolumn{2}{c}{Ungrouped}  & \multicolumn{2}{c}{Grouped}    \\ \cline{2-5} 
\rule{0pt}{2ex} & Circuits & Shots & Circuits & Shots \\ \hline
\rule{0pt}{4ex} Uniform spherical & \multirow{2.7}{*}{$O\left(N^2(\log_2d)^2d^2\right)$ \textsuperscript{a}} & $O\left(\frac{\gamma^2 N^4d^{8}}{\epsilon^2}\right)$ \textsuperscript{b} & \multirow{2.7}{*}{$O(Nd^2)$ \textsuperscript{c}}  & $O\left(\frac{\gamma^2 N^3(\log_2d)^4d^{6}}{\epsilon^2}\right)$ \textsuperscript{d} \\
\rule{0pt}{4ex} Pure uncoupled & & $O\left(\frac{\gamma^2 N^4 d^4}{\epsilon^2}\right)$ \textsuperscript{e} & & $O\left(\frac{\gamma^2 N^3 d^4}{\epsilon^2}\right)$ \textsuperscript{f} \\
\end{tabular}
\end{ruledtabular}
\footnotesize{For full expressions and proofs see Appendix \ref{app: shots}: \textsuperscript{a}Theorem~\ref{thrm: ungrouped circuits}, \textsuperscript{b}Theorem~\ref{thrm: scaling ungrouped shots}, \textsuperscript{c}Theorem~\ref{thrm: scaling}, 
\textsuperscript{d}Theorem~\ref{thrm: scaling shots}, 
\textsuperscript{e}Theorem~\ref{thrm: scaling shots zero ungrouped} and, \textsuperscript{f}Theorem~\ref{thrm:scaling shots zero}.}
\caption{Number of unique circuits required using our exact grouping as well as expected number of shots required to achieve absolute measurement error $\epsilon$ for a system of $N$ $d$-level one-dimensional QDOs. Here $\gamma$ is a typical scale for the inter-oscillator couplings. The expectations for number of shots are calculated with respect to distribution over the uniform spherical measure of states (Uniform spherical) as well as the pure state of the uncoupled QDOs (Pure uncoupled). For ungrouped results all Pauli operators in the Hamiltonian are measured separately, for grouped results Pauli operators are measured in groups according to the method described in Appendix~\ref{app: pauli decomp} Theorem~\ref{thrm: scaling}.}
\label{table: shots}
\end{table*} 

\begin{figure*}
\includegraphics{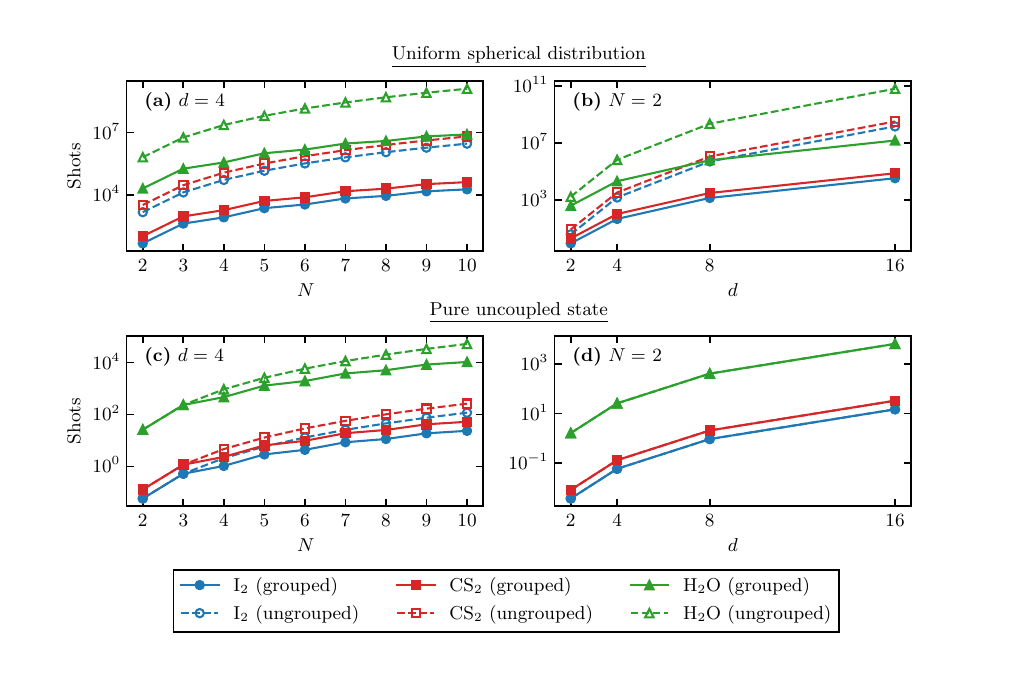}
\caption{Estimated number of shots required to measure $\langle \H(R_{\text{vdW}}) \rangle$ with error $\epsilon=10\%\times\Delta{E}(R_{\text{vdW}})$ for a range of different molecules. Here $R_{\text{vdW}}$ is the van der Waals diameter of the respective molecules. The expected number of shots is calculated with respect to the spherically uniform measure [\textbf{(a)} and \textbf{(b)}] as well as with the uncoupled state [\textbf{(c)} and \textbf{(d)}] with and without grouping commuting operators. The number of shots is shown both as a function $N$ for fixed $d=4$ and as a function of $d$ for fixed $N=2$. Note that in \textbf{(d)}, the grouped and ungrouped lines are the indistinguishable from one another. Values for $R_{\text{vdW}}$ and $\alpha$ are taken from Refs. \cite{maroulis1997electrooptical, yamamoto2006pressure, schropp2008polarizability, huang2013size}.}\label{fig: shots}
\end{figure*}

When combining the single oscillator non-interacting and two-oscillator coupling terms together to measure the total Hamiltonian, if all $\gamma_{i,j}$ are non-zero there will be a total of $(\log_2d)N + \frac{1}{8}d^2(\log_2d)^2N(N-1)$ non-identity Pauli operators. Naively $\H$ can be measured by measuring each of these operators using its own circuit, however the number of measurements required can be reduced significantly by measuring mutually-commuting operators concurrently. Many heuristic methods exist for partitioning the terms into commuting subsets, see for example Refs. \cite{jena2019pauli, crawford2021efficient, verteletskyi2020measurement, izmaylov2019unitary, zhao2020measurement, yen2020measuring,gokhale2020n}. From the structure of our Hamiltonian we can derive an exact method for grouping terms into subsets that qubit-wise commute and therefore obtain an upper bound on the number of unique measurements that need to be performed. We find an overall upper bound of $(N-1)\left(\frac{d}{2}\log_2d\right)^2 + 1$ or $N\left(\frac{d}{2}\log_2d\right)^2 + 1$ circuits, for even and odd $N$ respectively, needed to measure $\H$. In general, for a Hamiltonian with all-to-all $k$\textsuperscript{th} order coupling, there is an upper bound of $O\left((d\log_2d)^kN^{k-1}\right)$ scaling with the number of QDOs. Using heuristic sorting methods the scaling with respect to $d$ can be improved beyond that given by these upper bounds, however an exact scaling is not known. For a description of the grouping method, proof of these scalings as well as examples of the result of using heuristic grouping algorithms, see Appendix~\ref{app: pauli decomp}.

Since we expect that, for physically relevant systems, the ground states will be dominated by low-lying Fock states, we do not expect the need to increase $d$ significantly to model non-trivial systems. Therefore, we are primarily concerned with the scaling with respect to the number of QDOs $N$. Appendix~\ref{app: truncated} shows that this is the case for the pair of interacting QDOs and that $d=4$ is a good approximation to the true ground state until very large and unphysical couplings.

Furthermore, using this known grouping, we can estimate the expected number of shots needed to measure $\langle \H \rangle$ with absolute error $\epsilon$. Without knowing the ground state of an arbitrary Hamiltonian, we calculate the expected number of shots with respect to two different distributions of states in Hilbert space. The first distribution is the uniform spherical measure on the full Hilbert space which assumes no knowledge of the distribution of possible states~\cite[Chapter~7]{watrous2018theory}. The second is the pure state of the uncoupled ground state $\ket{\psi_0} = \ket{0}^{\otimes mN}$ which we expect to be close to the ground state of the system with small non-zero coupling. The scaling of the expected number of shots required for these two distributions is given in Table \ref{table: shots} for measuring all Pauli operators individually or by grouping them as described. We see that in both distributions, grouping the measurements improves the scaling from $O(N^4)$ to $O(N^3)$ with respect to $N$ and significantly reduces the polynomial scaling with respect to $d$ in the case of the uniform spherical measure. In all cases, the exact values as well as proofs for the number of shots required can be found in Appendix~\ref{app: shots} (the number of shots is an upper bound in the case of the uniform spherical measure and tight for the pure uncoupled state).

To demonstrate the number of shots required for modelling real molecules of interest, Fig.~\ref{fig: shots} highlights the number of shots required to achieve a measurement error on $\langle \H \rangle$ of $\epsilon = 10\%\times\Delta E(R_{\text{vdW}})$. Here we consider $\Delta E(R_{\text{vdW}})$ to be the dispersion interaction between a pair of one-dimensional QDOs representing the molecule evaluated at $R_{\text{vdW}}$ the Van der Waals diameter of the molecules. Assuming distributions of quantum states of the uniform spherical measure as well as a pure uncoupled state, we see that the grouping method derived in this work can reduce the required number of shots by multiple orders of magnitude. Further reductions may be possible if heuristic methods are used to find more optimal groups.

For the prototypical orbital-based VQE methods there are in general $O(N^4)$ Pauli operators that must be measured where $N$ is the size of the molecular orbital basis set used. The best known scaling for the number of commuting operator collections is $O(N^3)$~\cite{gokhale2020n} which, when compared to the order $O(N)$ scaling in terms of the number of QDOs for the dipole coupled Hamiltonian here, suggests that the system we consider in this work may be significantly more scalable than those considered in orbital-based VQE methods. To our knowledge, no attempt has been made to find a scaling for the number of shots required for these.


\subsection{Variational ansatz} \label{subsec: ansatz}

For the results presented in this paper, we use a variational ansatz consisting of layers of two-qubit operators acting on nearest-neighbour qubits. Since the problem Hamiltonian is positive, the ground state will be real-valued (up to a global phase) and we restrict ourselves to using real-valued, two-qubit operators based on the optimal construction of general $
\text{SO}(4)$ operators found by Vatan and Williams~\cite{vatan2004optimal}. We removed two parameterised $R_z$ gates from their construction to reduce the number of variational parameters. This naturally restricts the available operations implemented by each block to a subset of $\text{SO}(4)$ but was found to still give reasonable results in simulation and experiment while reducing the number of circuits required for gradient measurements. The circuit used for two QDOs is shown in Fig.~\ref{fig: circuit}a which has a total of twelve variational parameters, the two-qubit circuit blocks are given in Fig.~\ref{fig: circuit}b. 

\begin{figure}
\includegraphics{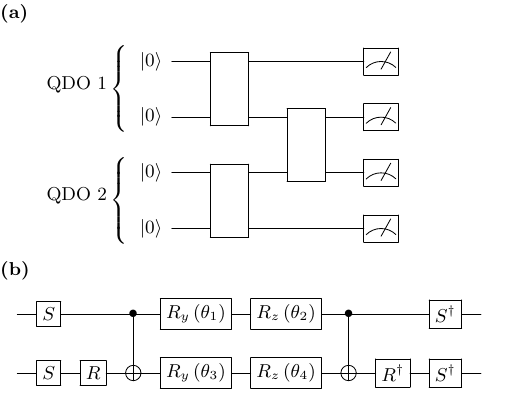}
\caption{Circuits used in experiments. \textbf{(a)} Four qubit variational ansatz representing the state of two $d=4$ QDOs. The upper (lower) qubit within the two-qubit register for each QDO corresponds to the most (least) significant bit in the binary encoding of the corresponding QDO Fock states. Each two-qubit block consists of a four-parameter real-valued operator.
\textbf{(b)} Composition of the four-parameter real-valued operators. $R_z$ and $R_y$ gates in the centre are Pauli Z and Y rotation gates, each one contains a single variational parameter. The remaining gates are fixed with $S=R_z\left(\frac{\pi}{2}\right)$ and $R=R_y\left(\frac{\pi}{2}\right)$. This block is based on the general SO(4) operator from Ref.~\cite{vatan2004optimal} with two parameterised $R_z$ gates removed.}
\label{fig: circuit}
\end{figure}

Although more costly than ans\"atze designed to make efficient use of the quantum hardware~\cite{kandala2017hardware}, quantum advantage may be achieved for ab-initio VQE by using chemistry-inspired ans\"atze such as unitary coupled cluster (UCC)~\cite{romero2018strategies}. Analogous ans\"atze for bosonic systems can be constructed~\cite{mcardle2019digital,ollitrault2020hardware,majland2021resource}, namely the Unitary-vibrational-coupled-cluster (UVCC) ansatz (following the naming convention of McArdle et al.~\cite{mcardle2019digital}). This is a unitary extension to the vibrational coupled cluster (VCC) used in classical computations~\cite{christiansen2004second,christiansen2004vibrational} and, unlike the VCC ansatz, the UVCC ansatz obeys the variational principle and may deal better with strong static correlations~\cite{mcardle2019digital}.

The UVCC ansatz will lead to significantly larger circuits than our ansatz in Fig.~\ref{fig: circuit} due to the non-trivial decomposition of bosonic operators into Pauli operators that can be implemented natively on quantum devices, thus we do not consider its application in this work. The analysis of commutation relations presented in Appendix~\ref{app: pauli decomp} as well as various encoding methods~\cite{majland2021resource} will be important in understanding the experimental cost of the Trotterisation required for such an ansatz.


\section{Experimental results using IBM quantum computers}\label{sec: ibmq dispersion}

\begin{figure*}
\includegraphics{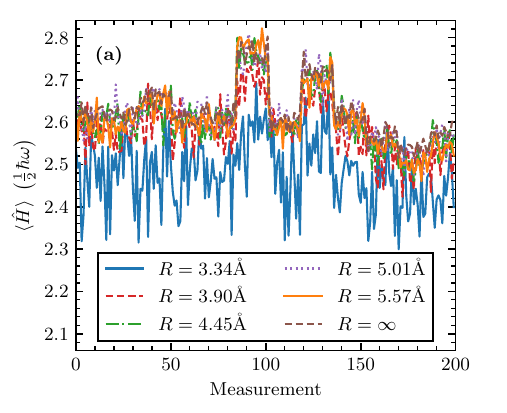}
\includegraphics{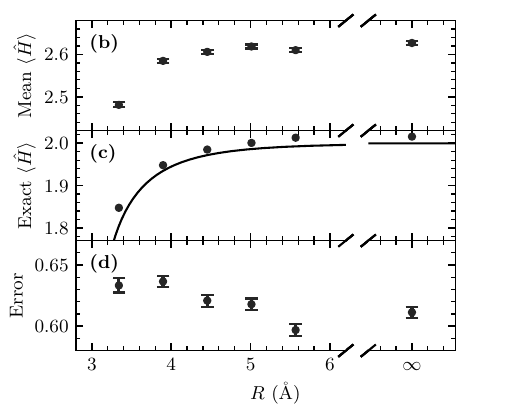}
\caption{Experimental measurements used for error subtraction. \textbf{(a)} Measurements of $\langle \H(R)\rangle_{\rho(\bm{\theta}_R)}$ where $\bm{\theta}_R$ are the optimal parameters found for each separation $R$ evaluated on \textit{ibm\_lagos} quantum computer. As described in the main text, each of the 200 experiments is performed for each value of $R$ before moving onto the next one. \textbf{(b)} Mean values of the measurements for each $R$. \textbf{(c)} Exact calculation of $\H$ with no noise for the optimised parameters. Solid curve shows the true ground state. \textbf{(d)} Absolute difference between exact value of $\langle \H \rangle$ and mean of noisy measurements. This shows that, although the effect of noise varies with each experiment, by measuring each $R$ in the order described, the effect of noise is approximately equal for each data point. Data points at $R=\infty$ correspond to the uncoupled Hamiltonian and parameters found using the optimisation procedure.
}\label{fig: subtraction}
\end{figure*}

To demonstrate the feasibility of calculating weak intermolecular interactions from the coarse-grained model we represent two linear, non-polar molecules at distance $R$ apart using two one-dimensional QDOs. We choose model parameters matching those of \Itwo, which is known to produce a dimer molecular liquid above its melting point of 114\degree C in which dispersion forces dominate the cohesion.

For each QDO, we use a $d=4$ level Fock basis, leading to a total of four qubits required for the two-molecule system. We model the pair of \Itwo molecules, both aligned with their molecular axis along the intermolecular axis such that the one-dimensional QDO models their induced polarisation along the molecular axis. We use a value for the polarisabilty of $\alpha=\alpha_{zz}=14.5 \text{\AA}^3$ the polarisability coefficient along the molecular axis~\cite{maroulis1997electrooptical} and a harmonic trapping frequency of $\hbar\omega=9.61\text{eV}$. For the coarse-grained electronic model, it is usual to choose the $\omega$ parameter (or alternatively the point charge $q$ and mass $\mu$ if working with raw Drude particle parameters). This is done to capture a certain chosen property, such as the $C_6$ London dispersion coefficient, and then further terms (such as $C_8$ and $C_9$ terms) arise naturally from the model. In this case, we choose $\hbar\omega$ such that the full three-dimensional QDO in this setup would give $C_6$ coefficient matching experimental measurements of \Itwo~\cite{kendrick2012empirical}. For the pair of interacting molecules, the full three-dimensional model and the total London dispersion energy for the pair of interacting molecules can be found from repetitions of the one-dimensional model using the correct possibilities for each remaining axis a discussion of which is deferred until  Sec.~\ref{subsec: 3d}. As proof of principle, we consider polarisability only along the intermolecular axis. To generate the dispersion energy as a function of inter-molecular separation $R$, we perform the VQE procedure to calculate the ground state energy of the coupled QDOs for a range of possible coupling constants.

For the experimental optimisation procedure we used, for $R=\infty$,  the quantum computer \textit{ibmq\_santiago} and, for all other values of $R$, \textit{ibmq\_montreal} (For further details about the hardware used in this work, see Appendix~\ref{app: hardware}). Each evaluation of $\langle\H\rangle$ is performed by running ten circuits with 8192 shots each where the 21 Pauli operators needed to estimate $\langle\H\rangle$ are collected into ten groups of operators that qubit-wise commute as described in Sec.~\ref{subsec: algorithm rep}. We evaluated gradients using the parameter shift rule~\cite{mitarai2018quantum} and performed the optimisation using the ADAM optimiser (learning rate 0.25, decay rates 0.9, 0.99)~\cite{kingma2014adam}. For each value of $R$, the VQE procedure was run for 200 optimisation steps with the final set of parameters taken as the output. For each experiment, the values of $\bm{\theta}$ are independently randomly initialised to values taken uniformly on the interval $[-\pi/2,\pi/2]$.

Once the optimal set of parameters was found, we extracted the binding energy from the noisy quantum computer using a simple scheme to subtract the effect of device noise. The binding energy due to the London dispersion interaction at a certain $R$ is then given by 
\begin{equation}
\Delta E(R) = E(R) - E(R=\infty),
\end{equation} where $E(R)$ is the ground state energy at $R$ and with $E(R=\infty)$ corresponding to the ground state energy of the completely uncoupled state. By modelling the output of the noisy quantum computers as a unitary channel depending on variational parameters $\bm{\theta}$ along with depolarising noise with error rate $\lambda$, the output state for a given set of variational parameters is

\begin{equation}
\rho(\bm{\theta}) = (1-\lambda)\ketbra{\psi(\bm{\theta})}{\psi(\bm{\theta})} + \frac{\lambda}{2^{mN}}\I,
\end{equation}

where $\ket{\psi{(\bm{\theta})}}$ is the output state from a noiseless ansatz with parameters $\bm{\theta}$. Assuming $\ket{\psi{(\bm{\theta})}}$ is the required ground state, this corresponds to a noisy expectation value for the Hamiltonian

\begin{equation}
\langle \hat{H}(R)\rangle_{\rho(\bm{\theta})} = (1-\lambda)E(R) + \frac{\lambda}{2^{mN}} \tr[\hat{H}(R)].
\end{equation}

Therefore, assuming $\lambda$ is independent of the variational parameters and does not change between experiments, the dispersion energy is given by

\begin{equation}\label{eqn: error subtraction}
\Delta E(R) = \frac{\langle \H(R)\rangle_{\rho(\bm{\theta})} - \langle \H(R=\infty)\rangle_{\rho(\bm{\theta_\infty})}}{(1-\lambda)},
\end{equation}

where $\bm{\theta}_\infty$ are variational parameters found for the ground state of $\H(R=\infty)$, here we have used that $\tr[\H(R)] = \tr[\H_\text{non-int}] = dN \; \forall R$. We use \eqref{eqn: error subtraction} to calculate the dispersion energy from noisy measurements, estimating $\epsilon$ by measuring the fidelity of the uncoupled ground state on the noisy computer $\mathcal{F}=\tr\left[\rho(\bm{\theta}_\infty)\ketbra{0}{0}^{\otimes mN}\right] = (1-\lambda)+\frac{\lambda}{2^{mN}}$. This can be done without additional cost by reusing the counts measured when calculating $\langle \H \rangle_{\rho(\bm{\theta}_\infty)}$. We find that, for the devices and circuits we use in this work, this method to subtract the device noise gives accurate values for the dispersion energy.

To measure the dispersion energy for the optimised states we use the \textit{ibm\_lagos} quantum computer. One practical consideration when performing this subtraction, is that due to the significant time dependence of the noise levels of the device, it is important to evaluate the expectations for different values of $R$ in~\eqref{eqn: error subtraction} with the same qubits on the same device as close to simultaneously as possible. For the results presented here we took the mean of 200 repetitions of each measurement of $\langle \H \rangle$ (with 8192 shots per circuit), performing the measurements for each $R$ sequentially before moving onto the next of the 200 repetitions. As many measurements were performed as possible in a single job and all jobs were submitted simultaneously to minimise the time between successive jobs. Fig.~\ref{fig: subtraction} shows how the level of noise varies significantly throughout the 200 evaluations but the drift seen in the expectation values is approximately matched for each of the different values of $R$ when measured in this way.

The spectrum achieved by the optimisation procedure in this proof of principle experiment is shown in Fig.~\ref{fig: spectrum}. From this, we see that by using the variational procedure combined with our ansatz, optimal parameters can be found using noisy gradients calculated on the real device that correspond to energy values close to the true ground state when evaluated without device noise. Furthermore, although the raw experimental energies are much larger than even the uncoupled energy of the system, using the simple subtraction method described above, we are able to successfully measure dispersion energies on the device closely matching the exact London dispersion interaction. From Fig.~\ref{fig: spectrum} we see that for all values of $R$ tested, the binding energy  using the quantum processor is within two standard deviations of the exact value. It is worth noting that the value of the ground state energy with the same variational parameters but evaluated using noiseless simulation (shown in Fig.~\ref{fig: subtraction}(b)) are all slightly above the true value given by the curve. This could be due to either imperfect or unconverged optimisation or systematic errors in the rotation gates applied by the device which are accounted for as over/under rotations in the variational parameters learnt by the algorithm and not reflected in the error-free simulation. We expect that these effects will be reduced by improvements in hardware performance and by using better gradient descent schemes or better fine tuning of the settings for the ADAM optimisation used.


\section{Extending the model}\label{sec: model extensions}

\subsection{Three-dimensional QDOs}\label{subsec: 3d}

As described in Sec.~\ref{subsec: coarse graining}, for a complete description of the dipole-dipole interaction and higher order terms, full three-dimensional QDOs are required. Here we demonstrate how our representation and variational algorithm can be extended to three-dimensional QDO systems.

The Hamiltonian for such a system is given by 
\begin{equation}\label{eqn: 3d hamiltonian}
\begin{split}
\H =& \sum_{i=1}^N \left(\frac{1}{2\mu}\bfp_i^2 + \frac{1}{2}\mu\omega^2\bfx_i^2\right)\\
&+ q^2\sum_{j>i}\frac{3(\bfx_i\cdot\bfR_{i,j})(\bfx_j\cdot\bfR_{i,j}) - \bfx_i\cdot\bfx_jR_{i,j}^2}{R_{i,j}^5},
\end{split}
\end{equation} 

where $\bfx_i=(\x_i,\y_i,\z_i)$ and $\bfp_i = (\p_{x_i},\p_{y_i},\p_{z_i})$ are the position and momentum operators for QDO $i$ in each of the three Cartesian coordinates. $\bfR_{i,j}$ is the vector separating the centres of the pairs of QDOs $i$ and $j$ and $\mu$, $\omega$ and $q$ represent the effective mass, natural frequency and charge for the QDO particles which, for simplicity, are assumed to be isotropic and equal for all oscillators.

To represent the state of these three-dimensional oscillators on a quantum computer, we treat each of the three spatial dimensions as a separate harmonic oscillator with position and momentum operators $\set{(\x,\p_x), (\y, \p_y), (\z,\p_z)}$. As in Sec.~\ref{subsec: algorithm rep}, each of these one-dimensional oscillators can be represented using a binary encoding in $m$ qubits. Thus, the state of each three dimensional oscillator state is encoded in a register of $3m$ qubits:
\begin{equation}
\begin{split}
&\mathcal{H}_{d-\text{Fock}}^{\otimes 3} \rightarrow \mathcal{H}_2^{\otimes m} \otimes \mathcal{H}_2^{\otimes m} \otimes \mathcal{H}_2^{\otimes m}:\\&\ket{\underline{n}_x,\underline{n}_y,\underline{n}_z} \rightarrow \ket{\text{bin}_m(n_x),\text{bin}_m(n_y),\text{bin}_m(n_z)}.
\end{split}
\end{equation}

Following Sec. \ref{subsec: algorithm rep}, an upper bound for the number of circuits that are required to measure $H$ is $(N-1)\left(\frac{d}{2}\log_2d\right)^2 + 1$ or $N\left(\frac{d}{2}\log_2d\right)^2 + 1$, for even and odd $N$ respectively. See Appendix \ref{app: pauli decomp} for details.

Extending the two oscillator system to three dimensions, the Hamiltonian simplifies to 
\begin{equation}
\H = \sum_{u\in\set{x,y,z}} (\hat{u}_1^2 + \hat{p}_{u_1}^2) + (\hat{u}_2^2 + \hat{p}_{u_2}^2) + \gamma_{uu}\hat{u}_1\hat{u}_2,
\end{equation}
where $\gamma_{xx} = \nicefrac{2\alpha_{xx}}{R^3}$, $\gamma_{yy} = \nicefrac{2\alpha_{yy}}{R^3}$ and $\gamma_{zz} = \nicefrac{-4\alpha_{zz}}{R^3}$, $\alpha_{xx}$, $\alpha_{yy}$, $\alpha_{zz}$ represent the dipole polarisability along each of the three axes (which we now allow to be anisotropic, corresponding to $q$ being an anisotropic tensor) and the intermolecular axis is in the $z$ direction; along which the physical separation is $R$. We have ignored dimensions to match the structure of \eqref{eqn: hamiltonian}.  We see from this that the Hamiltonian is equivalent to three independent one-dimensional oscillators with the Hamiltonian matching \eqref{eqn: hamiltonian} for $N=2$ with the appropriate choices of coupling constants. Therefore, in the previous section where we modelled the interaction of two linear molecules aligned along an axis, an extension to a full three dimensional system follows trivially.


\subsection{Induction, anharmonic and non-linear terms}\label{subsec: anharmonic}

The quadratic Hamiltonians that we have presented so far can easily be extended to include higher-order and anharmonic terms. Here we briefly describe three extensions of the model giving rise to additional interactions and phenomena. We also discuss how the decomposition and grouping strategies described in Sec.~\ref{subsec: shots} can be applied to these richer systems.

The cheapest and simplest modification is the inclusion of induction effects by applying a uniform electric field $E$ to the oscillators. This simple modification can, for two or more oscillators, lead to non-additive effects~\cite{cipcigan2019electronic}. For one-dimensional oscillators this leads to an interaction of the form $\H_\text{E-field} = \sum_{i=1}^N q_iE\x_i$ for $N$ Drude oscillators with charges $q_i$. These can be measured using the same decompositions for $\x$ as were used previously since these interactions are guaranteed to qubit-wise commute with terms of the form $\x \otimes \x$ incurring no additional circuits to measure the total Hamiltonian.

For more realistic induction models, we can instead consider oscillators in a Coulomb potential caused by one or more charges. In this case, the Hamiltonian will have the form $\H_\text{coulomb} = \sum_{i=0}^N\sum_{j=0}^M \phi_j(\x_i)$ for $N$ oscillators in Coulomb potentials represented by $\phi_j$ caused by $M$ fixed charges. To measure a Hamiltonian of this form, each of the Coulomb potentials must be Taylor expanded up to some finite order $w$:
\begin{equation}\label{eqn: coulomb}
    \phi_j(\x_i) \sim \frac{1}{|r_j-\x_i|} \approx \frac{1}{r_j} + \frac{\x_i}{r_j^2} + \frac{\x_i^2}{r_j^3} + \dots + \frac{\x_i^{(w-1)}}{r_j^w},
\end{equation}
where $r_j \gg |\x_i|$ is the position of the point charge relative to the centre of the QDO. Again, since each term of the form $\x_i^w$ acts on distinct oscillators they can all be measured simultaneously. The expected contribution to the number of shots will be $O(\sqrt{N})$ with respect to $N$ thus leaving the asymptotic scaling unchanged when added to the linear dipolar Hamiltonian of \eqref{eqn: hamiltonian}.

Using a similar addition to the Hamiltonian we are able to relax the assumption of a harmonic trapping potential and consider anharmonic trapping potentials that have not been accessible for large systems using classical computation methods~\cite{jones2013quantum}. Additional terms of the form $\H_\text{anharm}=\omega_3\x_i^3+\dots+\omega_w\x_i^w$ with appropriate model parameters defining coefficients $\set{\omega_3,\omega_4, \dots,\omega_w}$ can be used to choose a desired, potentially asymmetric, trapping potential. The measurement cost will scale in the same way as terms in \eqref{eqn: coulomb}. Relaxing the harmonic assumption is required for inclusion of some higher order terms (such as four-dipole hyperpolarisabilty~\cite{jones2013quantum,jones2010quantum}) and leads to non-linear optical effects such as the Kerr effect, second- and third- harmonic generation and, forms of optical mixing~\cite{boyd2020nonlinear}.

Finally, we introduce non-linear interaction terms. As an example we consider interactions of the form $\H_\text{non-lin} = \sum_{i,j,k=0}^N\beta_{\langle i,j,k\rangle}\x_i \x_j \x_k$. Terms similar to this have been considered for use in quantum algorithms for vibrational spectra, noting that these are difficult to deal with classical computational methods~\cite{sawaya2019quantum, mcardle2019digital}. If all $O(N^3)$ possible terms in $\H_\text{non-lin}$ are non-zero, using the exact grouping that we have described will contribute an $O(\frac{d^9N^5}{\epsilon^2})$ cost to the number of shots required. In practice the number of shots may be considerably lower as heuristic sorting methods will make use of additional qubit-wise commutation not considered in the grouping method we have used in this work. For higher order interactions, these commuting pairs will likely be more common. Also, many physical systems may not include all-to-all interactions or will have interactions that decay strongly with distance such that many of the interaction terms can be ignored. Terms of this form, as well as higher order non-linear interaction terms will allow for a large amount of flexibility in the model and will allow for the investigation of a wide range of physical and chemical phenomena.


\section{Further experimental results}\label{sec: further results}

\subsection{Anharmonic systems}\label{subsec: anharmonic experiments}

To test the performance of the algorithm for a system of anharmonic oscillators, we investigate two systems with 1D anharmonic QDOs on both simulated and real quantum computers.

\begin{figure}
\includegraphics{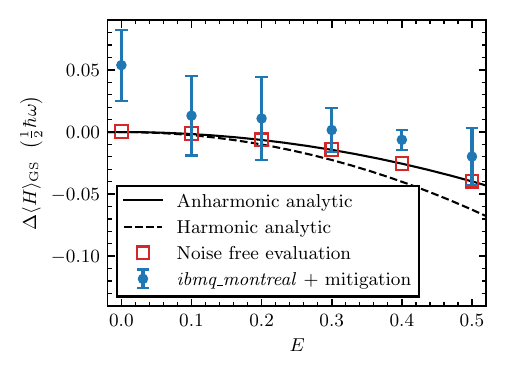}
\caption{Deviation from zero field ground state energy of 1D anharmonic oscillator in external field $E$. Blue circles correspond to evaluating the energy on the real \textit{ibmq\_montreal} quantum processor using error mitigation techniques described in the main text. Red squares correspond to evaluating the final energy with no device or shot noise. Both sets of data points use the same variational parameters which were found through optimisation performed on a simulated version of \textit{ibmq\_montreal}. Black solid line shows the exact spectrum found by numerically diagonalising the Hamiltonian. Black dashed line shows the exact spectrum for the analagous harmonic system (setting $x^4$ term to zero).}
\label{fig: anharmonic with field}
\end{figure}

We first consider the system of a single oscillator in an external potential, modifying the onsite potential to include a quartic term. In this case, the Hamiltonian is given by $\H = \x^2 + \p^2 + \frac{1}{5} \x^4 + E\x$, where parameter $E$ is a uniform external field which is varied. We used two qubits to give a truncated Fock space of $d=4$ and an ansatz consisting of the single four parameter block given in Fig.~\ref{fig: circuit}(b).

Optimisation was performed using a simulated version of the \textit{ibm\_montreal} device and using the same optimiser and settings as in Sec.~\ref{sec: ibmq dispersion} for 200 optimisation steps. As with the experiments in Sec~\ref{subsec: many oscillators}, for the first datapoint at $E=0$, $\bm{\theta}$ were initialised to values taken randomly on the uniform interval $[-\pi/2,\pi/2]$. Proceeding data points were then run in order of increasing $E$ initialising $\bm{\theta}$ to the optimised values linearly extrapolating the optimised parameters from the previous data points with respect to the two corresponding values of $E$.

Once the optimisation was completed, we evaluated the final energy value using the optimised parameters on the real \textit{ibm\_montreal} device to see how real device noise affected the accuracy of the energy measured. Since the system was no longer harmonic, the zero field ground state was not known to be the zero state and so the subtraction method given in Eqn.~\eqref{eqn: error subtraction} could not be used to remove device noise from the measured value. Instead two error mitigation techniques - zero noise extrapolation~\cite{temme2017error, li2017efficient} and one based on the Lanczos method~\cite{suchsland2021algorithmic} were used in combination to measure the final ground state energy. See Appendix~\ref{app: mitigation} for details on these two methods, as well as experimental settings used in this work.

The results of this procedure are shown in Fig.~\ref{fig: anharmonic with field}. We see that values of the ground state energy evaluated without noise agree well with the true value found by directly diagonalising the Hamiltonian matrix. As in previous experiments, this shows that, despite the simulated device noise in the optimisation procedure, a good set of variational parameters are found by the algorithm and that the ansatz is able to accurately generate the required groundstates. We see that, when combined with error mitigation, the interaction energy of the system shows some agreement with the true value, however the uncertainties in the values mean that we are not able to resolve the difference in interaction energies between the harmonic and anharmonic systems. We note that in this system, the energy resolution required is significantly smaller than the dispersion interaction of Sec.~\ref{sec: ibmq dispersion}.

An increase in the measurement uncertainty is a consequence of both error mitigation techniques. This uncertainty can be decreased, as well as the overall reduction in the error, by using a larger number of shots (thereby increasing runtime) or by improving error rates on the quantum computer.

We additionally perform experiments for a system of two linearly coupled oscillators (the system using in Sec.~\ref{sec: ibmq dispersion} with additional  quartic onsite potential terms). For such a system the Hamiltonian is given by $\H = \x_1^2 +\p_1^2 + \frac{1}{5}\x_1^4 + \x_2^2 + \p_2^2 + \frac{1}{5}\x_2^4 + \gamma\x_1\x_2$. We again use two qubits per oscillator and use the ansatz shown in Fig~\ref{fig: circuit}(a) and the same optimisation settings as in Sec~\ref{sec: ibmq dispersion}. Again, values of $\bm{\theta}$ were randomly initialised for $\gamma=0$ data point and initial values of $\bm{\theta}$ for remaining datapoints were chosen by extrapolating from previous optimised values. Optimisation was performed using a simulated version of the \textit{ibmq\_montreal} device and we calculated the final values for the ground state energy found using exact simulation - removing the effect of device or shot noise. The results of these experiments are given in Fig.~\ref{fig: anharmonic coupled}. These show a good agreement between the true ground state energy and results found through the optimisation procedure.

\begin{figure}
\includegraphics{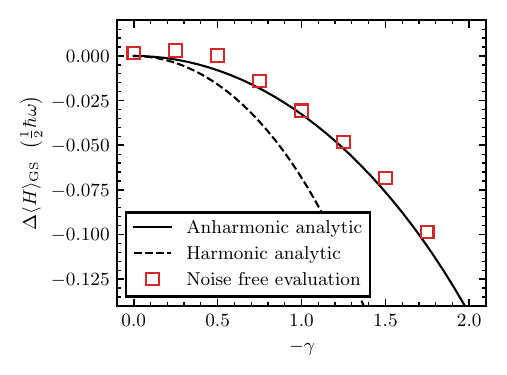}
\caption{Deviation from uncoupled ground state energy for system of two coupled anharmonic oscillators with optimisation performed on simulated \textit{ibmq\_montreal} device. Data points show ground state energy relative to true ground state energy of uncoupled system, evaluated without device or shot noise. Solid line gives exact spectrum of the anharmonic system calculated by direct matrix diagonalisation. Dashed line gives exact spectrum of the corresponding harmonic system.}
\label{fig: anharmonic coupled}
\end{figure}


\subsection{Many oscillators}\label{subsec: many oscillators}

We also test our variational algorithm for more than two-body systems, namely $N=3,4$ and $5$ oscillators. Using a Hamiltonian of the form Eq.~\eqref{eqn: hamiltonian} again, coupling constants $\gamma_{i,j}$ were chosen to represent systems of 1D QDOs arranged in regular $N$ sided polygons and allowed to oscillate perpendicular to the inter-oscillator plane. For these systems, $\gamma_{i,j} = 2\alpha/{R_{i,j}^3} \propto 2\alpha/D^3$ where $R_{i,j}$ are the distances between oscillators $i$ and $j$, $D$ gives the diameter of the circle that effectively encloses the $N$ vertex polygons. By varying the ratio of $2\alpha/D^3$ (effectively varying the size of the polygons) we generate the spectra shown in the results of Fig.~\ref{fig: many}.

For the results shown in Fig.~\ref{fig: many}, we simulate the performance of our VQA using a range of different ans\"atze depths. We use an exact simulation for the optimisation procedure and final energy measurement so there are no contributions from finite number of shots or device noise. As a variational ansatz, we use a similar construction to that of the ansatze used in Sec.~\ref{sec: ibmq dispersion} with the slight modification of using cyclic connectivity such that there are additional $\text{SO}(4)$ circuit blocks between the first and last qubits in even numbered layers. We represent each QDO with two qubits such that the local dimension of the oscillators is $d=4$ and experiments require six, eight and ten qubits for $N=3,4,$ and $5$ respectively. For each data point, 500 optimisation steps and the same optimisation settings as in Sec.~\ref{sec: ibmq dispersion}. For the $\alpha/D^3=0$ datapoints, $\bm{\theta}$ were initialised to values taken uniformly on the interval $[-\pi/2,\pi/2]$. Proceeding data points were then run in order of increasing $\alpha/D^3$ initialising $\bm{\theta}$ to the optimised values for the previous data point (in the case for the second data point) or by linearly extrapolating the optimised parameters from the previous data points with respect to the two corresponding values of $\alpha/D^3$ (for all other datapoints).

From the results in Fig.~\ref{fig: many}, we see that the variational algorithm is able to successfully recreate the energy spectrum of the three, four and five QDO systems. Of particular interest is the spectrum of the three oscillator case in which the three body contribution, similar to the full Axilrod-Teller interaction energy~\cite{axilrod1943interaction}, makes a significant contribution to the dispersion energy and is accurately recreated by the $d=4$ truncated Fock space as well as the variational algorithm In fact, across all three systems, over the coupling parameter values considered, the truncated spectrum corresponding to the truncated Fock space matches well the exact ground state energy calculated through symplectic methods (see Appendix~\ref{app: symplectic} for details on how this was calculated). As the system size increases, a deeper ansatz circuit with more variational parameters is needed to capture reasonable estimates of the groundstate energy.

\begin{figure*}
    \centering
    \includegraphics{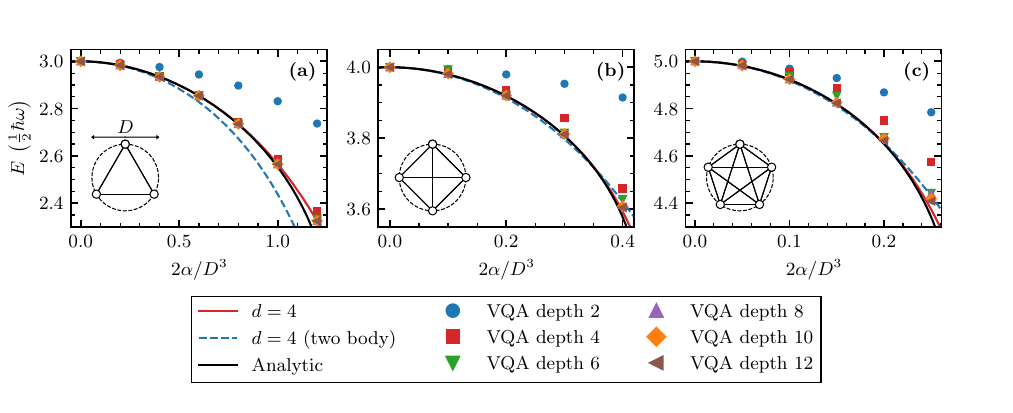}
    \caption{Simulation of VQA for systems of \textbf{(a)}~three, \textbf{(b)}~four and \textbf{(c)}~five QDOs arranged in regular polygons using ansatze of varying depth. Insets show respective oscillator geometries, labelling the diameter $D$ of the circle enclosing the polygons which was effectively varied to generate the spectra. Red solid line (labelled $d=4$) corresponds to numerically diagonalising Hamiltonian of oscillators of $d=4$ energy levels. Blue dashed line corresponds to the same system but only considering the two body pair-wise interactions. Black solid line (labelled analytic) corresponds to using symplectic representation of many oscillator system (as described in Appendix~\ref{app: symplectic}) which effectively corresponds to $d=\infty$. Data points correspond to outputs of VQA algorithm for varying depth ansatze.}
    \label{fig: many}
\end{figure*}


\section{Discussion and conclusions}\label{sec: discussion}

In this work we have developed a variational quantum algorithm for modelling intermolecular interactions using a coarse-grained framework. We have introduced an encoding scheme, variational ansatz and variational algorithm for digital quantum computers assessing measurement cost and scaling. We presented an exact method for grouping mutually commuting operators that can be measured simultaneously which gave bounds on the number of different circuits and number of shots required to measure a system with dipole-dipole interactions. We have shown evidence that this scaling may be more favourable than for existing VQE methods. We have also presented a pathway for adding anharmonic and higher order terms that are not easily accessible to classical computational methods. This, combined with the favourable scaling, suggests that our algorithm may be a more promising candidate for showing quantum advantage than other variational methods.

As with other applications of VQE methods, the global scaling of the algorithm presented - namely how many optimisation steps are required to find the ground state within a given accuracy - is not yet known. That the classical optimisation required for variational quantum algorithms is known to be NP-hard~\cite{bittel2021training} and the occurrence of barren plateaus~\cite{mcclean2018barren,cerezo2021cost,wang2021noise,marrero2021entanglement} - exponentially vanishing gradients that mean ansatze cannot be optimised efficiently - may limit the scalability of the variational approach. As seen in the field of classical machine learning however, worst case complexity theoretic arguments~\cite{blum1988training} do not necessarily preclude the use of such methods for practical applications, particularly as algorithmic improvements can be made to mitigate such issues. There is already a significant body of work on methods to address the issue of barren plateaus. These include investigating ansatz structure~\cite{pesah2021absence,bharti2021iterative, wiersema2020exploring} and parameter initialisation strategies~\cite{grant2019initialization,verdon2019learning,liu2021parameter} amongst others~\cite{anand2021natural,patti2021entanglement,sack2022avoiding}.  We expect that many methods found to avoid barren plateaus for VQAs in a general setting or for specific cost functions will likely be applicable to the work presented here. For the problem presented here, a more physically inspired ansatz (e.g., UVCC described in Sec.~\ref{subsec: ansatz}) may be needed to avoid barren plateaus induced by approximate 2-designs~\cite{mcclean2018barren} within the ansatz that we have used. Additionally, the cost function we use in this work is local and may be more resilient to barren plateaus than other applications of VQAs~\cite{cerezo2021cost,uvarov2021barren}.

We have experimentally demonstrated that van der Waals dispersion energies are resolvable on current generation quantum computing hardware and that they lie within physically realistic ranges found in molecular materials. This represents the first use of existing quantum computing capability to access the weak but influential intermolecular energy scales arising from non-covalent interactions (including correlated quantum fluctuations) relevant in condensed states of molecular assemblies. While the specific model used in this experiment can also be solved on classical processors, the impetus to develop and demonstrate an efficient and well characterised quantum implementation includes the following: 

Firstly, it expands the palette of problems that can can be accessed on quantum computing systems by providing a further level of electronic coarse graining while preserving a complete set of responses, fluctuations and correlation phenomena at long range. This methodology could therefore be combined with other electronic structure strategies allowing only selected electronic subsystems to be treated at the highest level of accuracy. Such combined hybrid approaches are likely to be required to treat complex or multiscale problems even if full CI computations can be realised in polynomial time complexity. Having established that dispersion interactions are detectable, extensions to the current implementation can be included to move up the interaction hierarchy to capture induction and many-body polarisation thereby opening opportunities for versatile molecular simulation problems on quantum architecture. 

Secondly, this framework may also be used to enrich existing simple models such as bead based descriptions of protein and branched polymer folding which can approach classically intractable combinatorial complexity and are already being actively explored for implementation in quantum architecture~\cite{fingerhuth2018quantum, robert2021resource}. Including coarse-grained representations of the long-ranged interactions will result in more realistic treatments of model folding problems.

Finally, we have shown that the aforementioned extension to the model, in which anharmonic on-site binding potentials can be introduced with low experimental overhead, can already be solved within our VQA and how the measured anharmonic energies are affected by device noise on superconducting processors. This paves the way for using quantum algorithms for anharmonic systems not easily accessible to classical algorithms and should lead to a better understanding of how anharmonic effects can be used to extend coarse-grained models for molecules as well as how new physics may arise in anharmonic many-body systems.


\begin{acknowledgments}
LWA would like to thank Christopher Willby for helpful discussions about symplectic methods for quadratic systems. LWA, MK, DJ acknowledge support from the EPSRC National Quantum Technology Hub in Networked Quantum Information Technology (EP/M013243/1) and the EPSRC Hub in Quantum Computing and Simulation (EP/T001062/1). MK and DJ acknowledge financial support from the National Research Foundation, Prime Ministers Office, Singapore, and the Ministry of Education, Singapore, under the Research Centres of Excellence program. The authors would like to acknowledge the use of the University of Oxford Advanced Research Computing (ARC) facility in carrying out this work~\cite{richards2015university}. This work was also supported by the Hartree National Centre for Digital Innovation programme, funded by the Department for Business, Energy and Industrial Strategy, United Kingdom.
\end{acknowledgments}

\appendix


\section{London dispersion in linear molecules and one-dimensional quantum Drude oscillators} \label{app: one-dimensional oscillators}

This appendix shows how the interaction between two axially symmetric linear molecules can be captured by one-dimensional QDOs. The results and derivations given here can be found in the literature but are combined and included here for the reader.

In the framework developed by London \cite{london1942centers}, the interaction between two anisotropic linear molecules that are axially symmetric about their respective principle intramolecular bond can be described by ellipsoids of charge. For a pair of identical linear molecules that share components of the dipole polarisability $\alpha_\parallel$ $(\alpha_\perp)$ and characteristic angular frequencies $\omega_\parallel$ $(\omega_\perp)$ parallel (perpendicular) to their intramolecular axis, the London Dispersion energy at large separation $R$ is 

\begin{widetext}
\begin{equation}
\begin{split}
    \Delta E = -\frac{1}{R^6}\bigg[&(\Cvv + \Chh - \Cvh) \left\{\sin^2(\theta_a)\sin^2(\theta_b)\cos(\phi_a-\phi_b)-2\cos(\theta_a)\cos(\theta_b)\right\}^2\\
    & + 3(\Cvh - \Chh)\left\{\cos^2(\theta_a) + \cos^2(\theta_b)\right\} + 2(\Cvh + 4 \Chh) \bigg], 
\end{split}
\end{equation}
\end{widetext}

where $\theta_a, \phi_a, \theta_b, \phi_b$ describe the orientation of molecules $a$ and $b$ relative to the intramolecular axis.  $C_{\cdot,\cdot}$ coefficients are given by
\begin{equation}
\begin{split}
    \Cvv = \frac{1}{8}\alpha_\parallel^2\hbar\omega_\parallel, \quad & \Chh = \frac{1}{8}\alpha_\perp^2\hbar\omega_{\perp},\\
    \Cvh = \frac{1}{4}\alpha_\parallel&\alpha_\perp\hbar\frac{\omega_\parallel\omega_\perp}{\omega_\parallel+\omega_\perp}.
\end{split}
\end{equation}
For further details about the geometry and derivation see Ref.~\cite{hirschfelder1964molecular}. When the two molecules are aligned parallel to each other and the intermolecular axis $\theta_a=\theta_b=0$ and $\phi_a=\phi_b$ (which can be set to $0$ without loss of generality) in the completely anisotropic limit $\alpha_{\parallel} \gg \alpha_\perp$ the London dispersion interaction is given by
\begin{equation}
\Delta E = -\frac{\alpha_\parallel^2\hbar\omega_\parallel}{2R^6}
\end{equation}

We now show how two interacting one-dimensional QDOs can recreate this dispersion interaction. The Hamiltonian for pair of one-dimensional QDO is
\begin{equation}
    \H = \frac{\p_1^2}{2\mu} + \frac{\p_2^2}{2\mu} + \frac{1}{2}\mu\omega^2(\x_1^2+\x_2^2 + \gamma \x_1 \x_2).
\end{equation}

where $\gamma = -4\frac{q^2}{\mu\omega^2}\frac{1}{R^3}$. Alternatively $\gamma = -4\frac{\alpha}{R^3}$ when written in terms of the dipole polarisability $\alpha=\frac{q^2}{\mu\omega^2}$. Defining the transformed variables
\begin{equation}
    \p_\pm = \frac{1}{\sqrt{2}}(\p_1 \pm \p_2), \qquad \x_\pm = \frac{1}{\sqrt{2}}(\x_1 \pm \x_2),
\end{equation}
the Hamiltonian can be rewritten in this new uncoupled basis as
\begin{equation}
\begin{split}
   \H =& \left(\frac{\p_+^2}{2\mu} + \frac{1}{2}\mu\omega^2\left(1+\frac{\gamma}{2}\right)\x_+^2\right)\\
    &+\left(\frac{\p_-^2}{2\mu} + \frac{1}{2}\mu\omega^2\left(1-\frac{\gamma}{2}\right)\x_-^2\right).
    \end{split}
\end{equation}
Therefore the ground state energy is given by
\begin{equation}
\begin{split}
    E &= \frac{1}{2}\hbar \omega \left(\sqrt{1+\frac{\gamma}{2}} + \sqrt{1-\frac{\gamma}{2}}\right) \\
    & = \frac{1}{2}\hbar \omega \left(2- \frac{\gamma^2}{16} + O(\gamma^4) \right).
\end{split}
\end{equation}
Using the definition $\gamma = -\frac{4\alpha}{R^3}$ gives the $R^6$ term corresponding for the oscillator pairs
\begin{equation}
    \Delta E = -\frac{1}{32} \gamma^2 \hbar\omega = -\frac{1}{2}\frac{\alpha^2}{R^6}\hbar\omega.
\end{equation}


\section{Truncated Fock space} \label{app: truncated}

For our one-dimensional QDO system, we want to be sure that using a finite dimensional Fock space for each QDO still gives a reasonable approximation to the ground state energy of the true infinite-dimensional system. Fig. \ref{fig: truncated} shows how the ground state energy of the system deviates for different values of the local QDO dimension $d$. We can see that even for $d=4$, the ground state energy is approximated well until very large couplings $\gamma$. For reference, $\gamma=2$ corresponds to a Hamiltonian that is no longer positive-semi definite, at this point the charged Drude particles dissociate from the nuclei and the system breaks down. For physically relevant systems, we expect to require value of $\gamma$ much lower than 2 and we can therefore limit ourselves to small $d$. For the real device experiments in the main text, the smallest separation $R$ considered corresponds to the largest coupling value of $\gamma=1.55$.

\begin{figure}
\includegraphics{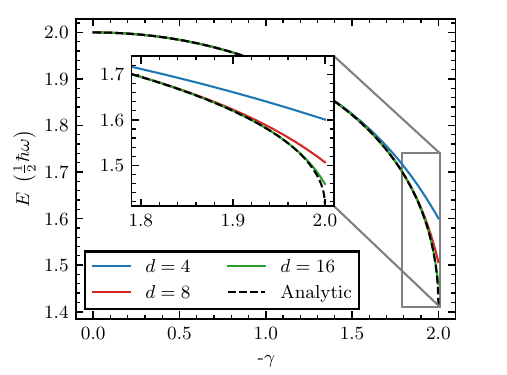}
\caption{Ground state energy of one-dimensional QDO pair as a function of coupling constant $\gamma$ for various values of $d$; the dimensionality of the truncated Fock space for each QDO. Analytic refers to exact analytic calculation of the ground state.}\label{fig: truncated}
\end{figure}


\section{Decomposition and measurement of Hamiltonian} \label{app: pauli decomp}

\subsection{Scaling in the total number of Pauli operators}

\begin{theorem}\label{thrm: ungrouped circuits}
For a $d$-level harmonic QDO, there are $\left(\frac{1}{2}d\log_2d\right)^2$ non-zero Pauli terms in the binary encoding of the $\x\otimes\x$ interaction.
\end{theorem}
\begin{proof}
The following single qubit operators can be represented in the Pauli basis
\begin{equation}\label{eqn: single qubit ops}
\begin{split}
\ket{0}\bra{0} & = \frac{1}{2}(\I + Z),\\
\ket{1}\bra{0} & = \frac{1}{2}(X - iY),\\
\end{split}
\qquad
\begin{split}
\ket{1}\bra{1} & = \frac{1}{2}(\I - Z),\\
\ket{0}\bra{1} & = \frac{1}{2}(X + iY).\\
\end{split}
\end{equation}

We first consider a $\x = \frac{1}{\sqrt{2}}(\aD + \a)$ for a single QDO, where the ladder operators $\aD$ and $\a$ for the $d$-level Fock space are given by
\begin{equation}
\begin{split}
\aD &= \sum_{n=0}^{d-1} \sqrt{n+1} \ket{\underline{n+1}}\bra{\underline{n}}, \\
\a &= \sum_{n=0}^{d-1} \sqrt{n+1} \ket{\underline{n}}\bra{\underline{n+1}},
\end{split}
\end{equation}
which follow the usual bosonic commutation relations. We begin by considering a single $\ket{\underline{n+1}}\bra{\underline{n}}$ which given in terms of $k$ single qubit operators is

\begin{equation}
\begin{split}
\ket{\underline{n+1}}\bra{\underline{n}} = & \left(\bigotimes_{i=k+1}^{m-1} \ket{\bin(n)_i}\bra{\bin(n)_i} \right)\\ &\otimes \ket{1}\bra{0} \otimes \left(\ket{0}\bra{1}\right)^{\otimes k},
\end{split}
\end{equation}

where $\bin(n)_i \in\set{0,1}$ is the binary digit at position $i$ of the binary representation of $n$ and $k(n) = \max \set{i \mid \bin(n)_i=0}$ is the position of the right-most $0$. Using the the definitions in \eqref{eqn: single qubit ops} this is equal to 
\begin{equation}\label{eqn: binary plus pauli ops}
\begin{split}
\ket{\underline{n+1}}\bra{\underline{n}} = & \left(\bigotimes_{i=k+1}^{m-1} (\I + (-1)^{\bin(n)_i}Z)\right)\\ &\otimes (X+iY) \otimes \left(X-iY\right)^{\otimes k}.
\end{split}
\end{equation}
Expanding \eqref{eqn: binary plus pauli ops} will give $2^m$ unique Pauli terms of the form $\hat{P} = c\sigma_i^{\otimes m}$ where $\sigma_i^{\otimes m} \in \set{\I, X, Y, Z}^m$ and $c \in \set{\pm 1, \pm i}$. Each $\hat{P}$ will either be hermitian or anti-hermitian, depending on whether $c=\pm 1$ or $c=\pm i$.

For each term in $\x$ like $(\ket{\underline{n+1}}\bra{\underline{n}} + \ket{\underline{n}}\bra{\underline{n+1}})$, we note that $\ket{\underline{n+1}}\bra{\underline{n}} = \left(\ket{\underline{n}}\bra{\underline{n+1}}\right)^\dagger$ and their sum is Hermitian. Therefore the Pauli terms within this summed term are simply two times the Hermitian terms in the expansion of \eqref{eqn: binary plus pauli ops} (anti-Hermitian terms cancel). Of the $2^m$ different $\hat{P}$ within \eqref{eqn: binary plus pauli ops}, half will be Hermitian ($c=\pm 1$) and half will be anti-Hermitian ($c=\pm i$). Therefore, there are a total of $\frac{1}{2}2^m$ unique Pauli operators in the compact representation of $(\ket{\underline{n+1}}\bra{\underline{n}} + \ket{\underline{n}}\bra{\underline{n+1}})$.

Finally, in the summation 
\begin{equation}\label{eqn: x pauli}
\hat{x} = \sum_{n=0}^{d-1} \sqrt{n+1} (\ket{\underline{n+1}}\bra{\underline{n}} + \ket{\underline{n}}\bra{\underline{n+1}}),
\end{equation}
there are $\log_2d=m$ unique values of $k(n)$, each  of which will contribute $\frac{2^m}{2}$ unique Pauli operators to the binary encoding of $\x$. Thus there are a total of $\left(\frac{1}{2}d\log_2d\right)^2$ Pauli operators within $\x\otimes\x$, where we have used $m=\log_2d$.
\end{proof}

\subsection{Scaling of Number of Circuits Required}

\begin{lemma}\label{lem: coupling groups}
The $\left(\frac{1}{2}d\log_2d\right)^2$ Pauli operators for the $\x\otimes \x$ interaction in the binary encoding can be measured using $(d-1)^2$ circuits.
\end{lemma}

\begin{proof}
We define the cover of an $M$ qubit Pauli operator $\P\in\set{\I,X,Y,Z}^M$ as 
\begin{equation}
\begin{split}
\cover(\P) = \big\{\hat{B}&\in\set{\I,X,Y,Z}^{\otimes M} \\ 
&| \quad \hat{B}^{(i)} =\P^{(i)} \text{ when } \P^{(i)} \neq \I\big\},
\end{split}
\end{equation}
where $\P^{(i)}$ is Pauli operator acting on the $i$\textsuperscript{th} qubit. For a collection of Pauli operators $\set{\P_1,\P_2,\P_3, \dots}$ that qubit-wise commute, there exists at least one basis $\hat{B}$ which covers all of the operators:
\begin{equation}
\hat{B} \in \cover(\P_1) \cup \cover(\P_2) \cup \cover(\P_3) \cup \dots.
\end{equation}

The expectation values of every operator in $\set{\P_1,\P_2,\P_3, \dots}$ can be measured using a single circuit in which qubit $i$ is measured in the $\hat{B}^{(i)}$ basis, where the $\pm 1$ measurement value is used if $\P^{(i)} \neq \I$, and discarded otherwise. Here we called bases independent if they do not qubit wise commute.

We show that there are $(d-1)$ independent covers for terms in $\x$ and therefore $(d-1)^2$ covers needed to measure $\x\otimes\x$.

We begin by considering the expansion of $\ket{\underline{n+1}}\bra{\underline{n}}$ given in \eqref{eqn: binary plus pauli ops}. Expanding out the $\bigotimes_{i=k+1}^{m-1} (\I + (-1)^{\bin(n)_i}Z)$ term on the first $(m-k-1)$ qubits will give tensor products of only $\I$ and $Z$. The covers for this $2^{(m-k-1)}$ operator subspace will always contain the basis $\bigotimes_{i=k+1}^{m-1} Z^{(i)}$, which can be used to measure them all. Then, expanding the $(X+iY) \otimes \left(X-iY\right)^{\otimes k}$ term on the remaining $(m-k+1)$ qubits, gives $2^{(k+1)}$ terms that do not qubit-wise commute and so each correspond to a unique basis.

Thus, there are $2^{(k+1)}$ independent bases for the terms in $\ket{\underline{n+1}}\bra{\underline{n}}$. Using the argument that $\ket{\underline{n+1}}\bra{\underline{n}} =  \ket{\underline{n}}\bra{\underline{n+1}}^\dagger$ as in Theorem \ref{thrm: scaling}, only the Hermitian terms remain when summing the two. This halves the number of unique covers and so $2^k$ bases are required for $\ket{\underline{n+1}}\bra{\underline{n}}+ \ket{\underline{n}}\bra{\underline{n+1}}$.

Again, for the sum of these terms given in \eqref{eqn: x pauli}, there are $m$ unique values for $k(n)\in \set{0,1,\dots,m-1}$ when written in the form of \eqref{eqn: single qubit ops}. For each unique value of $k(n)$, the value of $\bin(n)$ does not affect the coverings required [since the first $(m-k-1)$ qubits are always covered with $Z^{\otimes(m-k-1)}$], so each unique $k(n)$ contributes a unique set of $2^k$ coverings to the summation. Therefore, the total number of coverings needed for the Pauli terms in $\x$ are $\sum_{k=0}^{m-1}2^k = 2^m-1$, where we have used a standard result for the summation.

Finally, the number of unique coverings for the $\x\otimes\x$ interaction is then $\left(2^m-1\right)^2 = (d-1)^2$.
\end{proof}

\begin{definition}
[Hypergraphs \cite{voloshin2009introduction}]
A hypergraph is a generalisation of a graph in which edges can connect arbitrary number (i.e., not just a pair) of vertices. Formally, a hypergraph $G = (V,E)$ consists of $V$ a finite set of vertices and $E$ a set of non-empty subsets of $V$ called hyperedges. 
\end{definition}

\begin{definition}
[Complete Hypergraphs \cite{voloshin2009introduction}]
For $0\leq r \leq N$, a complete hypergraph $K_N^r=(V,E)$ is a hypergraph with vertices $V$ and set of hyperedges $E$ such that $|V|=N$ and $E$ contains all $r$-subsets of $V$. 
\end{definition}

\begin{lemma}[Baranyai's Theorem~\cite{van2001course}] For a complete hypergraph $K_N^r$ where $2\leq r<N$ and $r$ divides $N$, the $N$ vertices can be partitioned into subsets (called 1-factorisation) of $r$ vertices in ${{N-1} \choose {r-1}}$ different ways, such that each $r$-element subset (hyper edge) appears in exactly one partition.\label{thrm: baranyai}
\end{lemma}

\begin{proof} See \cite[Chapter 38]{van2001course}.\end{proof}

For example, the complete hypergraph $K_4^2$ contains four vertices, each connected by two-element hyperedges (i.e, traditional edges). Baranyai's theorem states that the vertices can be partitioned into ${{4-1} \choose {2-1}} = 3$ partitions. $K_4^2$ and its corresponding partitions are shown in Fig.~\ref{fig: baranyai}.

\begin{figure}
\centering

\includegraphics{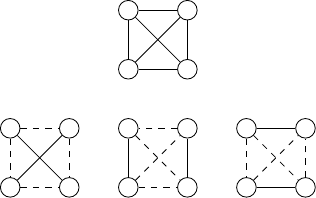}
\caption{Fully connected graph $K^2_4$ along with the three corresponding partitions of edges.}\label{fig: baranyai}
\end{figure}

\begin{theorem}\label{thrm: scaling}
The number of circuits required for measurement of the Hamiltonian  of the form
\begin{equation}
\H = \sum_{i=1}^N (\x_i^2 + \p_i^2) + \sum_{j>i} \gamma_{i,j} \x_i\x_j,
\end{equation}
for $N\geq2$ $d$-level QDOs in the binary encoding is upper bounded by
\begin{equation}
\begin{alignedat}{2}
    &(d-1)^2(N-1)+1; \quad &&N \text{ even}\\
    &(d-1)^2N+1;  &&N \text{ odd},
\end{alignedat}\label{eqn: 1D scaling}
\end{equation}
which is $O(d^2N)$.
\end{theorem}

\begin{proof}
All $(\x^2+p^2) = \sum_{n=0}^{d-1}n\ket{\underline{n}}\bra{\underline{n}}$ terms within the Hamiltonian can be decomposed into Pauli operators given in \eqref{eqn: nonint paulis} of the main text. All of these terms commute within each QDO register as well between QDO registers and can thus be measured using the single covering $Z^{\otimes Nm}$.\\

For the pair-wise coupling terms, from Theorem \ref{lem: coupling groups}, a single $\x\otimes\x$ term between two QDOs can be measured with $(d-1)^2$ circuits. The bases required to measure pairwise coupling terms between disjoint pairs of QDOs are guaranteed to commute since they act on different qubit registers. Therefore, by measuring each of the bases concurrently for collections of mutually disjoint interactions, the number of measurement bases can be reduced in a predictable way. 

We represent interaction terms as a graph in which each vertex corresponds to one of the $N$ QDOs and edges between vertices correspond represent a non-zero pairwise coupling. For all-to-all interactions this will be a fully connected graph. For even $N$ the number of partitions needed to separate interaction edges into disjoint sets is equal by Baranyai's theorem (Lemma~\ref{thrm: baranyai}) for the graph $K_N^2$. For odd $N$ the edges cannot be partitioned exactly - there will always be one unpaired vertex when collecting disjoint edges. Considering a dummy vertex where an edge between that and an QDO vertex gives no contribution to the Hamiltonian, the number of partitions is given by Baranyai's theorem for $K_{N+1}^2$. 

Using these collections, the number of bases that will need to be measured of the pairwise interaction terms is 

\begin{equation}
\begin{alignedat}{2}
    (d-1)^2{N-1 \choose 1} &= (d-1)^2(N-1); \quad &&N \text{ even},\\
    (d-1)^2{N \choose 1} &=(d-1)^2N; \quad &&N \text{ odd}.
\end{alignedat}
\end{equation}

Combining this with the single basis needed to measure the non-interaction term gives an upper bound for the number of bases that must be measured and thus the result.
\end{proof}

\begin{figure*}
\includegraphics{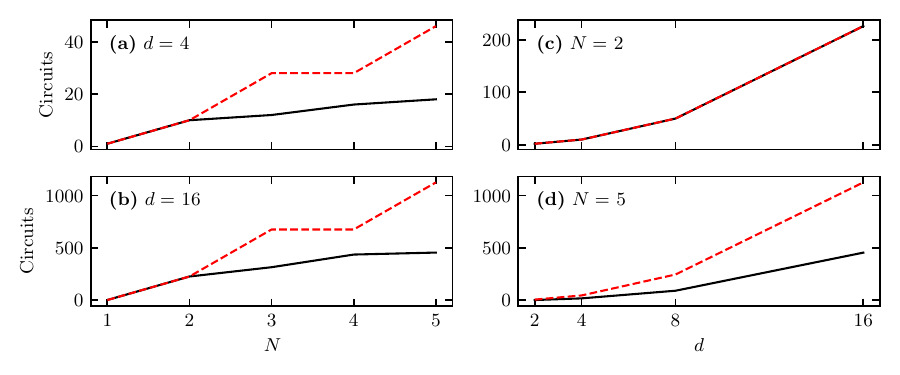}
\caption{Number of circuits required to measure Hamiltonian as \textbf{(a)} and \textbf{(b)}: a function of $N$ for two fixed values of $d$ as well as \textbf{(c)} and \textbf{(d)}: a function of $d$ for two fixed values of $N$. Red dashed line: upper bound scaling given by \eqref{eqn: 1D scaling}, black solid line: using largest-degree first colouring for grouping Paulis. In \textbf{(c)}, the two lines coincide. Trivially, the number of circuits for $N=1$ is always 1.}\label{fig: scaling}
\end{figure*}

This upper bound is plotted as a function of $N$ and $d$ in Fig. \ref{fig: scaling} as well as the number of circuits required when using a heuristic sorting algorithm, namely largest-degree first colouring ~\cite{verteletskyi2020measurement}. We see that for $N=1$ both curves agree as there are no interaction terms and that the heuristic sorted saturates the upper bound for $N=2$ QDOs. For $N>2$ we see that the sorting significantly reduces the number of circuits below the upper bound due to additional qubit wise commutation between non-disjoint interaction terms.

\begin{corollary}
The number of circuits required for measurement of the Hamiltonian for $N\geq r$ $d$-level QDOs with $r$-order coupling in the binary encoding is $O(d^2N^{r-1})$.
\end{corollary}

\begin{proof}
The proof is the same as that for Theorem \ref{thrm: scaling} except using $r>2$ for the $\x^{\otimes r}$ interaction in Lemma \ref{lem: coupling groups} and the corresponding hypergraph $K^r_N$ in Lemma \ref{thrm: baranyai}. This gives an upper bound on the total number of groups of
\begin{equation}
(d-1)^r{N-1 \choose r-1}+1,
\end{equation}
assuming $r$ divides $N$, which is $O(d^rN^{(r-1)})$. If $r$ does not divide $N$, the $O(d^rN^{(r-1)})$ scaling will remain as dummy vertices can be used to give the number of vertices that is divisible for $r$.

\end{proof}

\subsection{Number of circuits for three-dimensional QDOs}

\begin{theorem}
The number of circuits required for measurement of the Hamiltonian of the form
\begin{equation}\label{eqn: 3d hamiltonian 2}
\begin{split}
H = \sum_{i=1}^N & \left(\frac{\bf\p_i^2}{2\mu} + 2\mu\omega\bfx_i^2\right) \\
& + q^2\sum_{j>i}\frac{3(\bfx_i\cdot\bfR_{i,j})(\bfx_j\cdot\bfR_{i,j}) - \bfx_i\cdot\bfx_jR_{i,j}^2}{R_{i,j}^5},
\end{split}
\end{equation} 
for $N$ three-dimensional QDOs with $d$ levels per spatial dimension in the binary encoding is upper bounded by

\begin{equation}
\begin{alignedat}{2}
    &(d-1)^2(N-1)+1; \quad &&N \text{ even},\\
    &(d-1)^2N+1; &&N \text{ odd}.
\end{alignedat}
\end{equation}
\end{theorem}

\begin{proof}
The uncoupled terms within (\ref{eqn: 3d hamiltonian 2}) will correspond to operators $\sum_{n=0}^{2^m-1}n\ket{\underline{n}}\bra{\underline{n}}$ for each of the three spatial dimensions for each QDO. As before, these require just $Z$ measurements on each of the qubits and can thus be measured using a single circuit.\\

Written in terms of operators for each one-dimensional component, the interaction part of (\ref{eqn: 3d hamiltonian 2}) will contain all nine terms of the form $\set{\x_i,\y_i,\z_i}\otimes\set{\x_j,\y_j,\z_j}$, where across each $3m$ qubit register for a single QDO, the operators making up $\bfx_i$ are given by $\x_i = \x\otimes\I^{\otimes m}\otimes\I^{\otimes m}$, $\y_i = \I^{\otimes m}\otimes\x\otimes\I^{\otimes m}$ and $\z_i = \I^{\otimes m}\otimes\I^{\otimes m}\otimes\x$, following a similar pattern for $\bfx_j$, $\bf\p_i$ and $\bf\p_j$. Since every single term of the form $\x\x$, $\x\y$, $\y\z$ etc. in the coupling interactions will consist of the same decomposition of Pauli operators either acting on the same or different QDO subregisters. These all qubit wise commute and each term each of the $(\frac{d}{2}\log_2 d)^2$ Pauli terms for all oscillator-oscillator interaction terms can be measured together. Thus for each pair of coupled three-dimensional QDOs, only $(\frac{d}{2}\log_2 d)^2$ circuits are needed; the same number required for a pairwise interaction between one-dimensional QDOs.

Finally, using Lemma \ref{thrm: baranyai} and following the same logic as Theorem \ref{thrm: scaling}, an upper bound for the number of circuits needed to measure in \eqref{eqn: 3d hamiltonian 2} is $(N-1)\left(\frac{d}{2}\log_2d\right)^2 + 1$ or $N\left(\frac{d}{2}\log_2d\right)^2 + 1$, for even and odd $N$ respectively.
\end{proof}


\section{Shot requirements} \label{app: shots}
In this section, we derive an upper bound for an estimate of the number of shots required to achieve a given statistical error on the measurement of $\langle\H\rangle$ for a system of 1D coupled QDOs using the grouping of measurements described in Theorem \ref{thrm: scaling}.

We write a Hamiltonian in the form

\begin{equation}
\H = \sum_{g\in\mathcal{G}}\sum_{i \in g} a_i \P_i,
\end{equation}
where $\mathcal{G} = \set{g|[\P_i,\P_j]=0, \forall i,j \in g}$ is the set of all collections for the Pauli operatos $\P_i$ in which all operators in a collection qubit wise commute. A standard method \cite{wecker2015progress} using Lagrange multipliers to optimally distribute measurements between circuits gives a total number experimental shots of

\begin{equation}
S = \frac{1}{\epsilon^2}\left[\sum_{g\in\mathcal{G}} \sqrt{\sum_{i \in g} |a_i|^2\var{\P_i}}\right]^2
\end{equation}

to achieve an expected measurement error of $\epsilon$, assuming that $\cov[\P_i,\P_j]=0$ if $i\neq j$. The variance of a Pauli operator is given by $\var{\P_i} = 1-\langle \P_i \rangle$. In the absence of knowledge about the state that the expectation value is take with respect to, we take the expectation of this variance with respect to a spherical measure (See~\cite[Chapter~7]{watrous2018theory} and following a similar method to Ref.~\cite{crawford2021efficient}):

\begin{equation}\label{eqn: spherical var}
\mathbb{E}[\var \P_i]=
\begin{cases}0;\quad &\P_i=\I\\
1-\frac{1}{2^{n_i}+1};\quad &\text{otherwise}
\end{cases}
\end{equation}

where $n_i$ is the number of non-identity Pauli operators in $\P_i$. Additionally, in the uniform spherical measure, $\mathbb{E}[\cov(\P_i,\P_j)]=0$ if $i\neq j$.  Thus the expected number of shots $S$ to achieve error $\epsilon$ is upper bounded by

\begin{equation}\label{eqn: general shots}
\mathbb{E}[S] = \frac{1}{\epsilon^2}\left[\sum_{g\in\mathcal{G}} \sqrt{\sum_{i \in g} |a_i|^2\left(1-\frac{1}{2^{n_i}+1}\right)}\right]^2.
\end{equation}

\begin{theorem}[Expected shots over uniform measure when grouping measurements]\label{thrm: scaling shots}
Given a Hamiltonian

\begin{equation}
\H = \sum_{i=1}^N (\x_i^2 + \p_i^2) + \sum_{j>i} \gamma_{i,j} \x_i\x_j,
\end{equation}

for $N\geq2$ $d$-level QDOs in the binary encoding and measuring comuting operators according to the procedure described in Theorem \ref{thrm: scaling}, the expected number of shots $S$ required to achieve a constant measurement error $\epsilon$ is upper bounded by

\begin{widetext}
\begin{equation}
S\leq \frac{1}{\epsilon^2}\left[\frac{\sqrt{2}}{3}\sqrt{N}\sqrt{d^2-1} 
+ \left(\sum_{F\in\partition (G)}\sqrt{\sum_{\langle i,j\rangle\in F}\gamma_{i,j}^2}\right) 
\left(2d-\left[4+2\log_2 d\right]d^2+\left[2+2\log_2 d+\frac{1}{2}(\log_2 d)^2\right]d^3\right)\right]^2,
\end{equation}
\end{widetext}

when taken over the spherically symmetric distribution of states, where $\partition (G)$ is the set of coverings for a given 1-factorisation of the interaction graph $G$. Assuming $G$ is fully connected and that  values of $\gamma_{i,j}$ are independent of both $N$ and $d$ and randomly distributed between different $F\in \partition(G)$, then $S \in O\left(\frac{\gamma^2 N^3(\log_2d)^4d^6}{\epsilon^2}\right)$, where $\gamma$ is a typical scale for $\gamma_{i,j}$.
\end{theorem}

\begin{proof} We consider the groupings described in Theorem \ref{thrm: scaling}. Since interacting and non-interacting parts of the Hamiltonian have different forms and groupings, we consider each one's contribution to the total shots separately:

\begin{equation}
\epsilon^2 S = \left(\epsilon \tilde{S}_{\text{non-int}}+\epsilon \tilde{S}_{\text{int}}\right)^2.
\end{equation}

Starting with the uncoupled part of the Hamiltonian:
\begin{equation}
\sum_{i=1}^N (\x_i^2 + \p_i^2) = \sum_{i=1}^N 2^m\I^{\otimes mN}-\sum_{j=0}^{m-1} 2^{j}Z^{(im+j)},
\end{equation}
where $Z^{(im+j)}$ is the Z operator acting on the $j$\textsuperscript{th} qubit of QDO $i$, all terms form a single collection which gives
\begin{equation}\label{eqn: non int shots}
\epsilon \tilde{S}_{\text{non-int}} = \sqrt{N\sum_{j=0}^{m-1}2^{2j} \var Z} =\sqrt{\frac{2}{9}N\left(2^{2m}-1\right)},
\end{equation}
where we have used that $\mathbb{E}[\var Z] = \frac{2}{3}$ from (\ref{eqn: spherical var}).

For the interacting terms, the collection to which a single Pauli operators given by (\ref{eqn: binary plus pauli ops}) (i.e., which  $g\in\mathcal{G}$ it is sorted into) is determined by three things:

\begin{enumerate}
\item Which of the $(N-1)$ (even $N$) or $N$ (odd $N$) 1-factorisations of the interaction graph the corresponding interaction term belongs to.
\item The value of $k(n)$, the position of the rightmost zero in the binary representation of integer $n$ labelling the Fock state, for \emph{each of the two} the QDOs involved in an interaction term.
\item For a given pair of $k(n)$, which of the possible $2^{k+1}$ of $X$ and $Y$ operators is applied to the right most $(k+1)$ qubits.
\end{enumerate}
As before, the anti-Hermitian terms in (\ref{eqn: binary plus pauli ops}) are not needed in the sum $\ket{\underline{n+1}}\bra{\underline{n}} + \ket{\underline{n}}\bra{\underline{n+1}}$.

Separating the terms within $\H$ according to these criteria (their contributions to the equation labelled with curly braces), the expected number of shots required is 

\begin{equation}
\begin{split}
\epsilon \tilde{S}_{\text{int}} = \frac{1}{2}
&\underbrace{\sum_{k_1,k_2=0}^{m-1}}
_{1.}\;
\underbrace{\sum_{F\in\partition (G)}}
_{2.}\; \\
&\underbrace{\sum_{\substack{\P_1^{(:k_1+1)}\in\set{X,Y}^{\otimes(k_1+1)}\\ \P_2^{(:k_2+1)}\in\set{X,Y}^{\otimes(k_2+1)}}}}
_{3.}
\Gamma(F,k_1,k_2),
\end{split}
\end{equation}
where $\partition (G)$ is the set of the interaction graph $G$. From Lemma~\ref{thrm: baranyai}, if $G$ is fully connected then $|G| = (N-1)$ or $N$ (for $N$ even or odd). We have defined $\mathcal{N}_{k_1}=\set{n|k(n)=k_1, n\in \mathbb{Z}, 0\leq n <2^m}$ as the set of integers satisfying a given $k(n)$ (similarly for $\mathcal{N}_{k_2}$) and $\P_1^{(:k_1+1)}$ is a $(k_1+1)$ qubit Pauli operator (similarly for $\P_2^{(:k_2+1)}$). The $1/2$ prefactor results from the cancellation of anti-hermitian terms which make up half of the operators. The value of $\Gamma$ is given by
\begin{align}
\Gamma(F,k_1,k_2)=&\Bigg[\sum_{\langle i,j \rangle \in F}
\sum_{\substack{n_1\in\mathcal{N}_{k_1}\\ n_2\in\mathcal{N}_{k_2}}}\;
\sum_{\substack{\P^{(k_1+1:m)}_1\in\set{\I,Z}^{\otimes(m-k_1-1)}\\ \P_2^{(k_2+1:m)}\in\set{\I,Z}^{\otimes(m-k_2-1)}}} \nonumber\\
&\left(2\gamma_{i,j}\sqrt{n_1+1}\sqrt{n_2+1}\right)^2 \mathbb{E}[\var{\P_1\otimes \P_2}]
\Bigg]^{\frac{1}{2}},
\end{align}

where, $\P^{(k_1+1:m)}_1$ is an $(m-(k_1+1))$ qubit Pauli operator and $\P_1 = \P_1^{(:k_1+1)}\otimes\P^{(k_1+1:m)}_1$ (similarly for $\P^{(k_2+1:m)}_2$). This is the equivalent of the square root term in \eqref{eqn: general shots}; the contribution to the variance from a single collection of commuting operators. In the above, $\langle i,j \rangle$ correspond to edges within 1-factorisation $F$, $n_1$ and $n_2$ correspond to Fock terms satisfying the requirement of $k(n)$; the position of the right-most X or Y operator. Using the expectation on the spherical measure from \eqref{eqn: spherical var} gives
\begin{equation}
\mathbb{E}[\var{\P_1\otimes \P_2}] = \left(1-\frac{1}{2^{k_1+k_2+\alpha+\beta}+1}\right)
\end{equation}
where integers $0<\alpha_1<m-k_1$ and $0<\alpha_2<m-k_2$ count the number of Z operators in $\P_1^{(k_1+1:m)}$ and $\P_2^{(k_2+1:m)}$ respectively.

By recognising that when using the spherical measure, many of the expected variances are equivalent and only depend on the total number of X,Y or Z operators, the expected number of shots simplifies to

\begin{equation}
\epsilon \tilde{S}_{\text{int}} = \frac{1}{2}
\sum_{k_1,k_2=0}^{m-1}\;
\sum_{F\in\partition (K_N^2)}\;2^{k_1+1}2^{k_2+1}
\Gamma(F,k_1,k_2)
\end{equation}
with 
\begin{equation}
\begin{split}
\Gamma(F,k_1,k_2)= \Bigg[&\sum_{\langle i,j \rangle \in F}
\sum_{\substack{n_1\in\mathcal{N}_{k_1}\\ n_2\in\mathcal{N}_{k_2}}}\;
\sum_{\alpha=0}^{m-k_1-1}\sum_{\beta=0}^{m-k_2-1}\;\\
&{m-k_1-1 \choose \alpha} {m-k_2-1 \choose \beta}\\
&\times\left(2\gamma_{i,j}\sqrt{n_1+1}\sqrt{n_2+1}\right)^2\\
&\times\left(1-\frac{1}{2^{k_1+k_2+\alpha+\beta}+1}\right)
\Bigg]^{\frac{1}{2}}.
\end{split}
\end{equation}

Now we aim to find a simple upper bound for $S$. We start by grouping terms

\begin{widetext}
\begin{align}\label{eqn: shots}
\epsilon \tilde{S}_{\text{int}} = \left(\sum_{F\in\partition (G)}\sqrt{\sum_{\langle i,j\rangle\in F}\gamma_{i,j}^2}\right)
\sum_{k_1,k_2=1}^m 
& 2^{k_1+1}2^{k_2+1}
\sqrt{\sum_{\substack{n_1\in\mathcal{N}_{k_1}\\ n_2\in\mathcal{N}_{k_2}}}(n_1+1)(n_2+1)}\nonumber\\
&\sqrt{ \sum_{\alpha=0}^{m-k_1-1}\sum_{\beta=0}^{m-k_2-1}\; {m-k_1-1 \choose \alpha} {m-k_2-1 \choose \beta} \left(1-\frac{1}{2^{k_1+k_2+\alpha+\beta}+1}\right)}.
\end{align}
\end{widetext}
We note that

\begin{equation}\label{eqn: simp 2}
\begin{split}
&\sum_{\substack{n_1\in\mathcal{N}_{k_1}\\ n_2\in\mathcal{N}_{k_2}}}(n_1+1)(n_2+1)\\
&=\left(2^{k_1}\sum_{\gamma=0}^{2^{m-k_1-1}}(2\gamma+1)\right)
\left(2^{k_2}\sum_{\delta=0}^{2^{m-k_2-1}}(2\delta+1)\right)\\
&=2^{-k_1-2}(2^{k_1+1}+2^m)^2 \cdot 2^{-k_2-2}(2^{k_2+1}+2^m)^2
\end{split}
\end{equation}

and

\begin{equation}\label{eqn: simp 1}
\begin{split}
&\sum_{\alpha=0}^{m-k_1-1}\sum_{\beta=0}^{m-k_2-1}\; {m-k_1-1 \choose \alpha} {m-k_2-1 \choose \beta}\\
&\hphantom{\sum_{\alpha=0}^{m-k_1-1}\sum_{\beta=0}^{m-k_2-1}\;}\times\left(1-\frac{1}{2^{k_1+k_2+\alpha+\beta}+1}\right)\\
&\leq \sum_{\alpha=0}^{m-k_1-1} {m-k_1-1 \choose \alpha} 
\sum_{\beta=0}^{m-k_2-1} {m-k_1-1 \choose \beta} \\
&\hphantom{\sum \leq}= 2^{m-k_1-1}\cdot2^{m-k_2-1},
\end{split}
\end{equation}

where we have used a standard result for the summation of binomial coefficients. Substituting \eqref{eqn: simp 2} and \eqref{eqn: simp 1} into \eqref{eqn: shots} gives

\begin{widetext}
\begin{equation}
\begin{split}
\epsilon S_{\text{int}} \leq& \left(\sum_{F\in\partition (G)}\sqrt{\sum_{\langle i,j\rangle\in F}\gamma_{i,j}^2}\right)
\sum_{k_1,k_2=0}^{m-1} 2^{m-1}\left(4\cdot 2^{k_1+k_2} +2^{m+1}\cdot 2^{k_1} + 2^{m+1}\cdot2^{k_2} + 2^{2m}\right)\\
&=\left(\sum_{F\in(G)}\sqrt{\sum_{\langle i,j\rangle\in F}\gamma_{i,j}^2}\right) 
\left(2d-\left[4+2\log_2 d\right]d^2+\left[2+2\log d +\frac{1}{2}(\log_2 d)^2\right]d^3\right),
\end{split}
\end{equation}
\end{widetext}

using standard results for the summations and substitution $d=2^m$. Combining this with (\ref{eqn: non int shots}) gives the result.
\end{proof}

\begin{theorem}[Expected shots over uniform measure when not grouping measurements.]\label{thrm: scaling ungrouped shots}
Given a Hamiltonian

\begin{equation}
\H = \sum_{i=1}^N (\x_i^2 + \p_i^2) + \sum_{j>i} \gamma_{i,j} \x_i\x_j,
\end{equation}

for $N\geq2$ $d$-level QDOs in the binary encoding and measuring each corresponding Pauli operator separately, the expected number of shots required to achieve a constant measurement error $\epsilon$ is given by the upper bound

\begin{widetext}
\begin{equation}
\begin{split}
S\leq \frac{1}{\epsilon^2}\Bigg[&N\frac{\sqrt{6}}{3}(d-1) \\
&+ 
\left(\sum_{\langle i,j \rangle \in G}|\gamma_{i,j}|\right)
\left(
\left[3+2\sqrt{2}\right]d^2
- \left[2+\sqrt{2}\right]d^{\frac{5}{2}}
- \left[\frac{7}{2} + 3\sqrt{2}\right]d^3
+ \left[1+\sqrt{2}\right]d^{\frac{7}{2}}
+ \left[\frac{3}{2} + \sqrt{2}\right]d^4
\right)\Bigg]^2,
\end{split}
\end{equation}
\end{widetext}
when taken over the spherically symmetric distribution of states, where $G$ is the interaction graph. Assuming $G$ is fully connected and, that  values of $\gamma_{i,j}$ are independent of both $N$ and $d$, then $S \in O\left(\frac{\gamma^2 N^4d^{8}}{\epsilon^2}\right)$ where $\gamma$ is a typical scale for $\gamma_{i,j}$.
\end{theorem}
\begin{proof}
As in Theorem \ref{thrm: scaling shots}, we consider contributions from non-interaction and interaction terms separately. Starting with the uncoupled part of the Hamiltonian

\begin{equation}\label{eqn: ungrouped non-int shots}
\epsilon \tilde{S}_{\text{non-int}} = N\sum_{j=0}^{m-1}\sqrt{2^{2j} \var Z} = N\frac{\sqrt{6}}{3}(2^m-1).
\end{equation}

For the interaction part of the Hamiltonian
\begin{widetext}
\begin{equation}
\begin{split}
\epsilon \tilde{S}_{\text{int}} 
&= \frac{1}{2}\sum_{\langle i,j \rangle \in G}
\sum_{k_1,k_2=0}^{m-1}\;
\sum_{\substack{\P_1 \in \set{X,Y}^{\otimes k_1+1}\otimes\set{\I,Z}^{\otimes{m-k_1-1}}\\\P_2 \in \set{X,Y}^{\otimes k_2+1}\otimes\set{\I,Z}^{\otimes{m-k_2-1}}}}\;
\left[\sum_{\substack{n_1\in\mathcal{N}_{k_1}\\ n_2\in\mathcal{N}_{k_2}}} 4\gamma_{i,j}^2(n_1+1)(n_2+1) \mathbb{E}(\var \P_1\otimes\P_2)
\right]^{\frac{1}{2}}\\
&=\sum_{\langle i,j \rangle \in G}
|\gamma_{i,j}|
\sum_{k_1,k_2=0}^{m-1}\;
2^{k_1+1}2^{k_2+1}
\sqrt{\sum_{\substack{n_1\in\mathcal{N}_{k_1}\\ n_2\in\mathcal{N}_{k_2}}}(n_1+1)(n_2+1)}\\
&\hphantom{
=\sum_{\langle i,j \rangle \in G}
|\gamma_{i,j}|\sum_{k_1,k_2=0}^{m-1}\;
}
\sum_{\alpha=0}^{m-k_1-1}\sum_{\beta=0}^{m-k_2-1}\; {m-k_1-1 \choose \alpha} {m-k_2-1 \choose \beta}
\sqrt{\left(1-\frac{1}{2^{k_1+k_2+\alpha+\beta}+1}\right)}.
\end{split}
\end{equation}

Using (\ref{eqn: simp 1}) and (\ref{eqn: simp 2}) again and writing $d=2^m$, we get the inequality

\begin{equation}
\begin{split}
\epsilon \tilde{S}_\text{int} \leq &
\sum_{\langle i,j \rangle \in G}
|\gamma_{i,j}|
\sum_{k_1,k_2=0}^{m-1}\;
\left(2^{k_1+1}\cdot 2^{k_2+1}\right) \left(2^{m-k_1-1}\cdot 2^{m-k_2-1}\right)
\left(2^{-\frac{k_1}{2}-1}(2^{k_1+1}+2^m)\cdot2^{-\frac{k_2}{2}-1}(2^{k_2+1}+2^m)\right)\\
&=
\left(\sum_{\langle i,j \rangle \in G}|\gamma_{i,j}|\right)
\left(
\left[3+2\sqrt{2}\right]d^2
- \left[2+\sqrt{2}\right]d^{\frac{5}{2}}
- \left[\frac{7}{2} + 3\sqrt{2}\right]d^3
+ \left[1+\sqrt{2}\right]d^{\frac{7}{2}}
+ \left[\frac{3}{2} + \sqrt{2}\right]d^4
\right).
\end{split}
\end{equation}
\end{widetext}
Combining this with (\ref{eqn: ungrouped non-int shots}) gives the result.
\end{proof}

The preceding results assume no knowledge about the states that are being measured and use an expectation over all of Hilbert space respect to the spherically symmetric measure. Here we make a second estimate for the number of shots required using the pure state $\ket{\psi_0}=\ket{0}^{\otimes mN}$. $\ket{\psi_0}$ corresponds to the product state of all QDOs in their lowest Fock state and is assumed to be reasonably close to the states that we seek using the VQE procedure.

When not using the spherical measure, contributions to the measurement error from the covarience of terms within the same collection is non-zero may be non-zero. In this case, the expected number of shots is given by

\begin{equation}
S = \frac{1}{\epsilon^2}\left[\sum_{g\in\mathcal{G}} \sqrt{\sum_{i,j \in g} a_i a_j\cov[\P_i, \P_j]}\right]^2.
\end{equation}

where $\cov[\P_i,\P_j] = \langle \P_i\P_j \rangle_{\ket{\psi_0}} - \langle \P_i \rangle_{\ket{\psi_0}} \langle\P_j \rangle_{\ket{\psi_0}}$ with subscripts indicating the state that expectation values are taken with respect to.


\begin{theorem} [Expected shots for uncoupled state when grouping measurements.] \label{thrm:scaling shots zero}
Given a Hamiltonian
\begin{equation}
\H = \sum_{i=1}^N (\x_i^2 + \p_i^2) + \sum_{j>i} \gamma_{i,j} \x_i\x_j,
\end{equation}
for $N\geq2$ $d$-level QDOs in the binary encoding and measuring commuting operators according to the procedure described in Theorem \ref{thrm: scaling}, the expected number of shots $S$ required to achieve a constant measurement error $\epsilon$ for state $\ket{0}^{\otimes mN}$ is upper bounded by
\begin{equation}
S=\frac{1}{\epsilon^2}\frac{1}{16}
\left(\sum_{F\in\partition (G)}\;
\sqrt{\sum_{\langle i,j \rangle \in F}
\gamma_{i,j}^2} \right)^2d^4,
\end{equation}
where $\partition (G)$ is the set of coverings for a given 1-factorisation of the interaction graph $G$. Assuming $G$ is fully connected and that  values of $\gamma_{i,j}$ are independent of both $N$ and $d$ and randomly distributed between different $F\in \partition(G)$, then $S \in O\left(\frac{\gamma^2 N^3 d^4}{\epsilon^2}\right)$, where $\gamma$ is a typical scale for $\gamma_{i,j}$.\end{theorem}

\begin{proof}
Again we consider non-interacting and interacting contributions separately. For the non interacting contribution, we note that for $\ket{\psi}=\ket{0}^{\otimes mN}$, $\langle Z^{(i)} \rangle = \langle Z^{(j)} \rangle =\langle Z^{(i)}Z^{(j)}\rangle= 1 \; \forall \; i,j$. Therefore, 

\begin{equation}
\epsilon \tilde S_{\text{non-int}} = \sqrt{N\sum_{i,j=0}^{m-1}2^i2^j\cov[Z^{(i)},Z^{(j)}]}=0.\label{eqn: grouped zero state nonint shots}
\end{equation}

The contribution from the coupling terms is
\begin{equation}
\begin{split}
\epsilon \tilde{S}_{\text{int}} = &\frac{1}{2}
\sum_{k_1,k_2=0}^{m-1}\;
\sum_{F\in\partition (K_N^2)}\\
&\sum_{\substack{\P_1^{(:k_1)}\in\set{X,Y}^{\otimes(k_1+1)}\\ \P_2^{(:k_2)}\in\set{X,Y}^{\otimes(k_2+1)}}}
\Gamma(F,k_1,k_2),\label{eqn: grouped zero state int shots}
\end{split}
\end{equation}
where $\partition (K_N^2)$ is the set of the $(N-1)$ or $N$ (for $N$ even or odd) partitions of the interaction graph $K_N^2$, $\mathcal{N}_{k_1}=\set{n|k(n)=k_1, n\in \mathbb{Z}, 0\leq n <2^m}$ is the set of integers satisfying a given $k(n)$ (similarly for $\mathcal{N}_{k_2}$) and $\P_1^{(:k_1+1)}$ is a $(k_1+1)$ qubit Pauli operator (similarly for $\P_2^{(:k_2+1)}$). $\Gamma$ is given by
\begin{equation}\label{eqn: grouped zero state gamma}
\begin{split}
&\Gamma(G,k_1,k_2)\\&=
\Bigg[\sum_{\langle i,j \rangle \in F}
\sum_{\langle i',j' \rangle \in F}
\sum_{\substack{n_1\in\mathcal{N}_{k_1}\\ n_2\in\mathcal{N}_{k_2}}}\;
\sum_{\substack{n_1'\in\mathcal{N}_{k_1}\\ n_2'\in\mathcal{N}_{k_2}}}\\
&\hphantom{=}2^2\gamma_{i,j}\gamma_{i',j'}\sqrt{n_1+1}\sqrt{n_2+1}\sqrt{n_1'+1}\sqrt{n_2'+1}\\
&\hphantom{=}\times\cov\left[\hat{A}_{\langle i,j \rangle}(k_1,k_2,n_1,n_2),\hat{A}_{\langle i',j' \rangle}(k_1,k_2,n_1',n_2')\right]
\Bigg]^{\frac{1}{2}}
\end{split}
\end{equation}
where
\begin{equation}
\begin{split}
&\hat{A}_{\langle i,j \rangle}(k_1,k_2,n_1,n_2) \\
&= \left(\P_1^{(:k_1)}
\otimes \bigotimes_{\alpha=k_1+1}^{m-1}(\I+(-1)^{\bin(n_1)_\alpha}Z)
\right)_{\text{QDO } i}\\
&\hphantom{=}\otimes\left(\P_2^{(:k_2)}
\otimes \bigotimes_{\beta=k_2+1}^{m-1}
(\I+(-1)^{\bin(n_2)_\beta}Z)
\right)_{\text{QDO } j}
\end{split}
\end{equation}
acts on the sub registers corresponding to QDOs $i$ and $j$. 

We note that $\P_1^{(:k_1)}\in \set{\pm1}\times\set{X,Y}^{\otimes(k_1+1)}$ (recalling only Hermitian terms remain in $\H$) and so $\langle \P_1^{(:k_1)} \rangle_{\ket{\psi_0}}=0\; \forall k_1$ and similarly for $\P_2^{(:k_2)}$, from which it follows 
\begin{equation}
\begin{split}
\langle i&,j \rangle \neq \langle i',j' \rangle
\implies \\
&\cov\left[\hat{A}_{\langle i,j \rangle}(k_1,k_2,n_1,n_2),\hat{A}_{\langle i',j' \rangle}(k_1,k_2,n_1',n_2')\right]=0
\end{split}
\end{equation}
and 
\begin{equation}
\begin{split}
&\cov\left[\hat{A}_{\langle i,j \rangle}(k_1,k_2,n_1,n_2),\hat{A}_{\langle i,j \rangle}(k_1,k_2,n_1',n_2')\right]\\
&=\left\langle \hat{A}_{\langle i,j \rangle}(k_1,k_2,n_1,n_2)\hat{A}_{\langle i,j \rangle}(k_1,k_2,n_1',n_2') \right\rangle_{\ket{\psi_0}}\\
&=\left\langle
\P_1^{(:k_1)}\P_1^{(:k_1)}
\right\rangle_{\ket{\psi_0}}\left\langle
\P_2^{(:k_2)}\P_2^{(:k_2)}\right\rangle_{\ket{\psi_0}}\times\\
&\hphantom{=}\left\langle
\bigotimes_{\alpha=k_1+1}^{m-1}(\I+(-1)^{\bin(n_1)_\alpha}Z)(\I+(-1)^{\bin(n_1')_\alpha}Z)
\right\rangle_{\ket{\psi_0}}\\
&\hphantom{=}\left\langle
\bigotimes_{\alpha=k_1+1}^{m-1}(\I+(-1)^{\bin(n_2)_\alpha}Z)(\I+(-1)^{\bin(n_2')_\alpha}Z)
\right\rangle_{\ket{\psi_0}},
\end{split}
\end{equation}
where all expectation values are with respect to $\ket{\psi_0}$. Trivially $\langle\P_1^{(:k_1)}\P_1^{(:k_1)}\rangle_{\ket{\psi_0}}=\langle\P_2^{(:k_2)}\P_2^{(:k_2)}\rangle_{\ket{\psi_0}}=1$. It can be seen that
\begin{equation}
\begin{split}
&\left\langle
\bigotimes_{\alpha=k_1+1}^{m-1}(\I+(-1)^{\bin(n_1)_\alpha}Z)(\I+(-1)^{\bin(n_1')_\alpha}Z)
\right\rangle_{\ket{\psi_0}} \\
&\qquad = \;
\begin{cases}
4^{m-k_1-1}, & \mbox{if } n_1=n_1'=0,\\
0,& \mbox{otherwise},
\end{cases}
\end{split}
\end{equation}
and similarly for $n_2$ and $n_2'$. Therefore
\begin{equation}\label{eqn: covariance zero state subterms}
\begin{split}
&\cov\left[\hat{A}_{\langle i,j \rangle}(k_1,k_2,n_1,n_2),\hat{A}_{\langle i',j' \rangle}(k_1,k_2,n_1',n_2')\right]\\
&= 2^{4m-2k_1-2k_2-4}\delta_{i,i'}\delta_{j,j'}\delta_{n_1,0}\delta_{n_2,0}\delta_{n_1',0}\delta_{n_2',0},
\end{split}
\end{equation}
where $\delta$ is the Kronecker delta. We note that $\delta_{n_1,0}=\delta_{n_1,0}\delta_{k_1,0}  \;\forall k_1:n_1 \in \mathcal{N}_{k_1}$ and similarly for $\delta_{n_2,0}$. Then, substituting this into (\ref{eqn: grouped zero state gamma}) gives
\begin{equation}
\Gamma(G,k_1,k_2) = \sum_{\langle i,j \rangle} 2^2 \gamma_{i,j}^22^{4m-k_1-k_2-4}\delta_{k_1,0}\delta_{k_2,0},
\end{equation}
which substituting into (\ref{eqn: grouped zero state int shots}) gives
\begin{equation}
\epsilon \tilde S_\text{int} = \frac{1}{4}\sum_{F \in \partition{G}}\sqrt{\sum_{\langle i,j \rangle \in F}\gamma_{i,j}^2d^4}, 
\end{equation}
from which the result follows.
\end{proof}

\begin{theorem}[Expected shots for uncoupled state when not grouping measurements.]\label{thrm: scaling shots zero ungrouped}
Given a Hamiltonian
\begin{equation}
\H = \sum_{i=1}^N (\x_i^2 + \p_i^2) + \sum_{j>i} \gamma_{i,j} \x_i\x_j,
\end{equation}
for $N\geq2$ $d$-level QDOs in the binary encoding and measuring each corresponding Pauli operator separately, the expected number of shots $S$ required to achieve a constant measurement error $\epsilon$ for state $\ket{0}^{\otimes mN}$ is upper bounded by
\begin{equation}
S=\frac{1}{\epsilon^2}\frac{1}{16}\left(\sum_{\langle i,j \rangle \in G} |\gamma_{i,j}| \right)^2 d^4,
\end{equation}
where $G$ is the interaction graph. Assuming $G$ is fully connected and that  values of $\gamma_{i,j}$ are independent of both $N$ and $d$ and randomly distributed between different $F\in \partition(G)$, then $S \in O\left(\frac{\gamma^2 N^4 d^4}{\epsilon^2}\right)$, where $\gamma$ is a typical scale for $\gamma_{i,j}$.\end{theorem}

\begin{proof}Again we consider non-interacting and interacting contributions separately. As in theorem \ref{thrm:scaling shots zero}, the uncoupled contribution is zero.

The contribution from the coupling terms is
\begin{equation}
\begin{split}
\epsilon \tilde S_{\text{int}} =& \frac{1}{2}\sum_{k_1,k_2=0}^{m-1}\sum_{\langle i,j \rangle \in G} \sum_{\substack{\P_1^{(:k_1)}\in\set{X,Y}^{\otimes(k_1+1)}\\ \P_2^{(:k_2)}\in\set{X,Y}^{\otimes(k_2+1)}}}\\
&\sum_{\substack{n_1\in\mathcal{N}_{k_1}\\ n_2\in\mathcal{N}_{k_2}}} \sqrt{2^2\gamma_{i,j}^2\var[\hat{A}_{\langle i,j \rangle}(k_1,k_2,n_1,n_2)}]
\end{split}
\end{equation}
where
\begin{align}
&\hat{A}_{\langle i,j \rangle}(k_1,k_2,n_1,n_2) \nonumber\\
&= \left(
\P_1^{(:k_1)}
\otimes \bigotimes_{\alpha=k_1+1}^{m-1}(\I+(-1)^{\bin(n_1)_\alpha}Z)
\right)_{\text{QDO } i}\nonumber\\
&\hphantom{=}\otimes\left( 
\P_2^{(:k_2)}
\otimes \bigotimes_{\beta=k_2+1}^{m-1}
(\I+(-1)^{\bin(n_2)_\beta}Z)
\right)_{\text{QDO } j}
\end{align}
acts on the sub registers corresponding to QDOs $i$ and $j$. Using \eqref{eqn: covariance zero state subterms}, this simplifies to
\begin{equation}
\epsilon \tilde S_{\text{int}} = \frac{1}{4}\left(\sum_{\langle i,j \rangle \in G} |\gamma_{i,j}|\right)d^2,
\end{equation}
from which the result follows. \end{proof}


\section{Noise mitigation for anharmonic experiments}\label{app: mitigation}

Unlike for the harmonic systems, for anharmonic systems it is not possible perform the subtraction method given in \eqref{eqn: error subtraction}. This is because we do not have a good way to estimate $\lambda$ as the ground state of the zero field system is not simply given by the zero Fock state. Therefore, we use more costly error mitigation techniques in order to reduce the effect of device noise in the final energy measurements. We choose to use a combination of two error mitigation schemes, namely zero-noise extrapolation~\cite{temme2017error, li2017efficient} and an algorithmic error mitigation scheme based on Lanczos extrapolation~\cite{suchsland2021algorithmic}.

To reduce the effect of noise on the measurement of $\langle \H \rangle$ using zero noise extrapolation, the expectation value is measured multiple times with the noise in the system artificially increased to various levels. The amount that the noise is amplified in the system is parameterised by a variable which we call $\Lambda$ such that we have values $\{\langle H \rangle_\Lambda\}$ evaluated at each chosen value of $\Lambda$. A curve is then fitted to the values of $\Lambda$ and $\langle \H \rangle_\Lambda$ which can then be extrapolated back to $\Lambda=0$ to hopefully remove the effect of noise on the system. Zero noise extrapolation requires no more qubits than the original system and can be performed fairly easily. The difficulty lies in finding a good method to amplify the noise to enough values of $\Lambda$ such that a good fit can be found, as well as choosing a good fitting function for your device noise model.

In our experiments we boost the device noise by replacing each two qubit CNOT gate by an odd number of CNOT gates. Since in the noiseless setting an odd number of CNOT gates has the same effect as a single CNOT gate, this has the effect of artificially boosting the two-qubit gate noise - which is expected to be the main source of noise in the superconducting system. We use three different values of $\Lambda$ which correspond to using one, three and five CNOTs in place of each CNOT in the circuit. A linear fit was found to match values of $\{\langle H \rangle_\Lambda\}$ well and used to extrapolate to the $\Lambda=0$ to achieve a zero noise estimate.

For the further mitigation based on the Lanczos algorithm, we follow the method described in Ref.~\cite{suchsland2021algorithmic} to perform an $m=2$ Krylov subspace expansion of the Hamiltonian. In the $m=2$ case, this requires measuring $\langle H \rangle$, $\langle\H^2\rangle$ and $\langle \H^3 \rangle$ on the noisy device to create the Krylov basis. Minimisation in this basis (performed by exact matrix diagonalisation) is then used to find a minimum value for $\langle \H \rangle$ which satisfies the variational principle.

The zero-noise extrapolation and Lanczos algorithm were combined as in Ref.~\cite{suchsland2021algorithmic} to give an increased level of noise reduction. First zero noise extrapolation was used on expectation values of $\langle \H \rangle$, $\langle \H^2 \rangle$ and $\langle \H^3 \rangle$ which were then used in the Lanczos mitigation method. As with measurements of the cost function $\langle H \rangle$ in the main text, the Pauli strings for each operator were grouped into collections that qubit-wise commuted and each group measured with approximately 16 million shots. Operators $\H$, $\H^2$ and $\H^3$ consist of 6, 7 and 8 Pauli strings which were grouped into 2, 3 and 3 commuting groups respectively.


\section{IBM Quantum Hardware}\label{app: hardware}
For practical reasons related to device availability and queue-times we used three different quantum processors to evaluate our algorithm. All three processors were provided by IBM via the Qiskit~\cite{Qiskit} open source framework. We used \textit{ibmq\_montreal} (27 qubits, Processor type Falcon r4), \textit{ibm\_lagos} (7 qubits, Processor type Falcon r5.11H) and \textit{ibmq\_santiago} (5 qubits, Processor type Falcon r4L). 
The reported quantum volume of the aforementioned devices is 128, 32 and 32 respectively.

The experiments in Sec.~\ref{sec: ibmq dispersion} use 4 qubits on every processor, namely qubits 13, 14, 16 and 19 in \textit{ibmq\_montreal}, qubits 1, 3, 4 and 5 in \textit{ibm\_lagos} and qubits 1, 2, 3 and 4 in \textit{ibmq\_santiago} following the coupling map show in Fig.~\ref{fig:coupling}. The experiments were performed in a timeframe of 2 months and the noise parameters (and qubit properties) were varied along the period of the experiments. The experiments in Sec.~\ref{subsec: anharmonic experiments} used qubits 14 and 16 on \textit{ibmq\_montreal} and were performed over a period of two days.

\begin{figure}[!htb]
\centering
\includegraphics{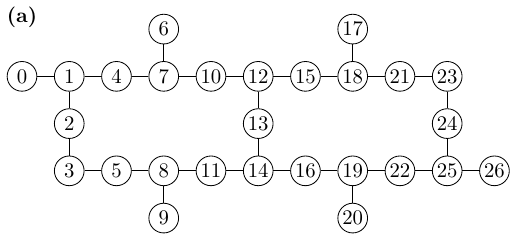}
\includegraphics{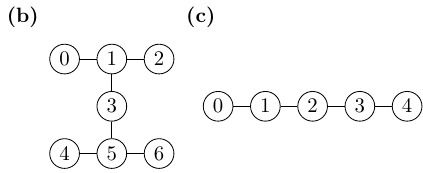}
    \caption{Coupling map of the IBM Quantum devices used to produce the results of the paper. The devices used are \textbf{(a)} \textit{ibmq\_montreal}, \textbf{(b)} \textit{ibm\_lagos} and \textbf{(c)} \textit{ibmq\_santiago} with 27, 7 and 5 qubits respectively.}
    \label{fig:coupling}
\end{figure}


\section{Symplectic eigenvalues of harmonic oscillators with linear coupling}\label{app: symplectic}
In the case of purely harmonic oscillators with linear coupling given by the Hamiltonian in Eqn.~\eqref{eqn: hamiltonian}, the exact ground state energy (without any truncation of the Fock space) can be found by calculating the eigenspectrum of a $2N\times2N$ matrix.

The Hamiltonian in~\eqref{eqn: hamiltonian} can be rewritten as 
\begin{equation}\label{eqn: quadratic hamiltonian}
    \H = \mathbf{\q}^\text{T} \mathbf{H} \mathbf{\q},
\end{equation}
where we have defined the vector operator $\mathbf{\q}^\text{T} = (\q_1,\q_2,\dots,\q_{2N-1},\q_{2N}) := (\x_1,\p_1,\dots,\x_N,\p_N)$ and correlation matrix $\mathbf{H}\in\mathbb{R}^{2N\times 2N}$ which has elements $H_{i,j}$ equal to the coefficients of the $\q_i\q_j$ operator in $\H$. In this representation, the bosonic commutation relation can be written as elements of a matrix $[\q_i,\q_j]= \Omega_{i,j}$ where
\begin{equation}
    \mathbf{\Omega} = \bigoplus_{k=1}^N \boldsymbol{\omega} \in \mathbb{R}^{2N\times2N}, \quad \boldsymbol{\omega} =
    \begin{pmatrix}
    0 & 1 \\
    -1 & 0
    \end{pmatrix}.
\end{equation}

The eigenspectrum of the matrix $|i\mathbf{\Omega}\mathbf{H}|$ (where the absolute value is defined in the operatorial sense~\cite{pirandola2009correlation}) is $\{\epsilon_1,\epsilon_2,\dots,\epsilon_{2N}\}$, where $\{\epsilon_i\}$ are known as the symplectic eigenvalues. The groundstate energy of the system is given by $E_0=\frac{1}{2}\sum_{i=1}^{2N} \epsilon_i$. We consider systems of increasing coupling strength up to the point at which $\min_i\epsilon_i = 0$ and the system becomes critical. For further details see~\cite{pirandola2009correlation, schuch2006quantum}.

We note that this method for calculating the groundstate energy only works for Hamiltonians that can be written in the quadratic form of~\eqref{eqn: quadratic hamiltonian} for which ground states and thermal states are Gaussian.


\bibliography{main.bib} 

\end{document}